\input ./style/arxiv-vmsta.cfg
\documentclass[numbers,compress,v1.0.1]{vmsta}

\usepackage{vmkcol}

\volume{4}
\issue{2}
\pubyear{2017}
\firstpage{141}
\lastpage{159}
\doi{10.15559/17-VMSTA80}

\startlocaldefs

\allowdisplaybreaks

\newtheorem{thm}{Theorem}
\newtheorem{lemma}{Lemma}

\theoremstyle{definition}

\newtheorem{assum}{Assumption}
\newtheorem{remark}{Remark}

\theoremstyle{proposition}
\newtheorem{propos}{Proposition}

\newtheorem{exam}{Example}

\endlocaldefs

\begin{document}
\begin{frontmatter}

\title{Bonus--malus systems with different claim types and varying deductibles}

%



\author{\inits{O.}\fnm{Olena}\snm{Ragulina}}\email{ragulina.olena@gmail.com}

\address{Taras Shevchenko National University of Kyiv, \\
Department of Probability Theory, Statistics and Actuarial Mathematics,
\\
64 Volodymyrska Str., 01601 Kyiv, Ukraine}

\markboth{O. Ragulina}{Bonus--malus systems with different claim types
and varying deductibles}

\begin{abstract}
The paper deals with bonus--malus systems with different claim types and
varying deductibles.
The premium relativities are softened for the policyholders who are in
the malus zone and these
policyholders are subject to per claim deductibles depending on their
levels in the bonus--malus scale
and the types of the reported claims. We introduce such bonus--malus
systems and study their basic
properties. In particular, we investigate when it is possible to
introduce varying deductibles,
what restrictions we have and how we can do this. Moreover, we deal
with the special case where
varying deductibles are applied to the claims reported by policyholders
occupying the highest
level in the bonus--malus scale and consider two allocation principles
for the deductibles.
Finally, numerical illustrations are presented.
\end{abstract}



\received{10 April 2017}
\revised{11 June 2017}
\accepted{13 June 2017}
\publishedonline{28 June 2017}
\end{frontmatter}

\section{Introduction and motivation}
\label{sec:1}

One of the main tasks of an actuary is to design a tariff structure
that fairly distributes the total risk
of potential losses among policyholders. To this end, he often has to
grade all policyholders into risk classes
such that all policyholders belonging to the same class pay the same premium.
Rating systems penalizing policyholders responsible for one or more
accidents by premium surcharges (or maluses),
and rewarding claim-free policyholders by giving them discounts (or
bonuses) are now in force in many developed
countries. Such systems, which are often called bonus--malus systems,
aim to assess individual risks better.

The amount of premium is adjusted each year on the basis of the
individual claims experience.
In practice, a bonus--malus scale consists of a finite number of levels
and each of the levels has its own relative
premium. After each year, the policyholder moves up or down according
to the transition rules and the number of
claims reported during the current year. Thus, bonus--malus systems also
encourage policyholders to be careful.
Note that the premium relativities are traditionally computed using a
quadratic loss function.
This method is proposed by Norberg in his pioneering work \cite{No1976}
on segmented tariffs.
Alternatively, Denuit and Dhaene \cite{DeDh2001} use an exponential
loss function to compute the relativities.

In most of the commercial bonus--malus systems used by insurance
companies, knowing the current level
and the number of claims during the current period suffices to
determine the next level in the scale.
Therefore, the future level depends only on the present and not on the past.
The numbers of claims in different years are usually assumed to be independent.
So the trajectory of each policyholder in the bonus--malus scale can be
considered as a Markov chain.
For details and more information concerning bonus--malus systems, we
refer the reader to
\cite{DeMaPiWa2007,Le1995,RoScScTe1999}. In particular, \cite
{DeMaPiWa2007} presents a comprehensive
treatment of the various experience rating systems and their
relationships with risk classification.

As pointed out by many authors, traditional bonus--malus systems suffer
from two considerable drawbacks:
\begin{enumerate}
\item[(i)]
The claim amounts are not taken into account. So a posteriori
corrections depend only on the number of claims.
In this case, policyholders who had accidents with small or large
claims are penalized unfairly in the same way.
In particular, this breeds bonus hunger when policyholders cover small
claims themselves in order to avoid
future premium increases.

\item[(ii)]
At any time the policyholders may leave the insurance company without
any further financial penalties.
Thus, bonus--malus systems create the possibility of malus evasion, i.e.
the situation when the policyholders
leave the insurance company to avoid premium increase because of
reported claims.
\end{enumerate}

An alternative approach, which, at least theoretically, eliminates the
second drawback, is proposed by Holtan \cite{Ho1994}.
He suggests the use of very high deductibles that may be borrowed by
the policyholders in the insurance company.
The deductibles are assumed to be constant for all policyholders, i.e.
independent of the level they occupy
in the bonus--malus system at the time of claim occurrence.
Although technically acceptable, this approach obviously causes
considerable practical problems.
Practical consequences of Holtan's proposal are investigated by Lemaire
and Zi \cite{LeZi1994}.
Particularly, it is shown that the introduction of high deductibles
increases the variability of payments
and the efficiency of the rating systems for most policyholders.

Bonus--malus systems with varying deductibles are introduced by
Pitrebois, Denuit and Walhin \cite{PiDeWa2005}.
Specifically, the a posteriori premium correction induced by the
bonus--malus system is replaced by a deductible
(in whole or in part).
To each level of the bonus--malus system in the malus zone is assigned
an amount of deductible, which is applied
to the claims filed during the coverage period either annually or claim
by claim.
Relative premiums at high levels of bonus--malus systems are often very
large and the systems can be softened
by introducing deductibles.
The insurance company compensates the reduced penalties in the malus
zone with the deductibles paid by policyholders
who report claims being in the malus zone.
This can be commercially attractive since the policyholders are
penalized only if they report claims in the future.\looseness=-1

As pointed out in \cite{PiDeWa2005}, combining bonus--malus systems with
varying deductibles presents a number of advantages.
Firstly, the policyholders will do all their best to prevent or at
least decrease the losses.
Secondly, even if the policyholder leaves the company after a claim, he
has to pay for the deductible.
Thirdly, relative premiums and amounts of deductibles may be tuned in
an optimal way in order to attract the policyholders.
The numerical illustrations show that the mixed case (reduced relative
premiums combined with per claim deductibles)
gives the best results. The amounts of deductibles are moderate in this case.
Nevertheless, this approach does not eliminate the first drawback
mentioned above.

To eliminate the first drawback, a few other approaches have been
proposed. Bonus--malus systems involving different
claim types are designed in \cite{PiDeWa2006}. Each claim type induces
a specific penalty for the policyholder.
In particular, 
claim amounts can be taken into account in this way.
For some early results in this direction, see also \cite{Le1995,Pi1997,Pi1998}.

The next approach is based on a generalization of results obtained by
Dionne and Vanasse \cite{DiVa1989,DiVa1992},
who propose a bonus--malus system that integrates a priori and a
posteriori information on an individual basis.
Specifically, the system is derived as a function of the years that the
policyholder is in the portfolio,
the number of claims and his individual characteristics.
Frangos and Vrontos \cite{FrVr2001} expand the frame developed in \cite
{DiVa1989,DiVa1992} and propose a generalized
bonus--malus system that takes into consideration simultaneously the
policyholder's characteristics, the number
of his claims and the exact amount of each claim. Particularly, it is
assumed that we have all information about
the claim frequency history and the claim amount history for each
policyholder for the time period he is in the portfolio.
In this generalized bonus--malus system, the premium is a function of
the years that the policyholder is in the portfolio,
his number of claims, the amount of each claim and the significant a
priori rating variables. Therefore, the future premium
depends on the past (all policyholder's history) and is actually
calculated individually for each policyholder.
This approach is extended and developed in \cite{MaHa2009,MeSa2005,TzVrFr2014}.

Another approach that takes into account claim amounts is proposed by
Bonsdorff \cite{Bo2005}.
The author considers a general framework for a bonus--malus system based
on the number of claims during the previous year
and the total amount of claims during the previous year.
The set of the bonus levels is some interval and the transitions
between the levels are determined by these characteristics.

G{\'o}mez-D{\'e}niz, Hern{\'a}ndez-Bastida and Fern{\'a}ndez-S{\'
a}nchez \cite{GoDeHeBaFeSa2014} obtain expressions
that can be used to compute bonus--malus premiums based on the
distribution of the total claim amount but not
on the claims which produce the amounts.

Another modification of traditional bonus--malus systems, which take
into account only the number of claims,
is considered by G{\'o}mez-D{\'e}niz in the recent paper \cite{GoDe2016}.
The author presents a statistical model, which distinguishes between
two different claim types, incorporating
a bivariate distribution based on the assumption of dependence.

The present paper deals with the case where both penalty types are used
for the policyholders who are in the malus zone.
Specifically, the premium relativities, which are computed using a
quadratic loss function, are softened and the policyholders
who are in the malus zone are subject to per claim deductibles. The
mixed bonus--malus system combining premium relativities
and deductibles is expected to be the most relevant in practice (see
\cite{DeMaPiWa2007,PiDeWa2005}).

We try to eliminate both drawbacks mentioned above. To take into
account claim amounts, we consider different
claim types and use the multi-event bonus--malus systems introduced in
\cite{PiDeWa2006}.
To eliminate the second drawback, we introduce varying deductibles for
the policyholders who are in the malus zone.
The deductibles depend on the level of the policyholder in the
bonus--malus scale and the types of the reported claims.
Such bonus--malus systems present a number of advantages and seem to be
very attractive for policyholders.
Namely, policyholders reporting small and large claims are not
penalized in the same way. This helps to avoid
or at least decrease bonus hunger. Moreover, all advantages mentioned
in \cite{PiDeWa2005} are also in force.

The rest of the paper is organized as follows. In Section~\ref{sec:2},
we describe the multi-event bonus--malus systems
introduced in \cite{PiDeWa2006}. In Section~\ref{sec:3}, we introduce
varying deductibles in such bonus--malus systems
and study basic properties of such systems. In particular, we
investigate when it is possible to introduce varying
deductibles in the bonus--malus systems with different claim types, what
restrictions we have and how we can do this.
Section~\ref{sec:4} deals with the special case where varying
deductibles are applied to the claims reported by
policyholders occupying the highest level in the bonus--malus scale. We
consider two allocation principles for
the deductibles, which seems to be natural and fair for policyholders.
In Section~\ref{sec:5}, we consider an example
of such a bonus--malus system, deal with exponentially distributed claim
sizes and present numerical illustrations.
Section~\ref{sec:6} completes the paper.

\section{Bonus--malus systems with different claim types}
\label{sec:2}

To take into consideration different claim types, we use the
multi-event bonus--malus systems introduced
in \cite{PiDeWa2006} (see also \cite{DeMaPiWa2007}).

Let us pick a policyholder at random from the portfolio and denote by
$N$ the number of claims reported by the policyholder
during the year. In what follows, we assume that there is no a priori
risk classification (or we work inside a specified
rating cell). Denote by $\lambda>0$ the a priori annual expected claim
frequency. Let $\varTheta$ be the (unknown) accident
proneness of this policyholder, i.e. $\varTheta$ represents the residual
effect of unobserved characteristics.
The risk profile of the portfolio is described by the distribution
function $F_{\varTheta}$ of $\varTheta$.
It is usually assumed that $\mathbb P[\varTheta\ge0]=1$ and $\mathbb
E[\varTheta]=1$.
Thus, the actual (unknown) annual claim frequency of this policyholder
is $\lambda\varTheta$.\looseness=-1

We assume that the number of claims $N$ is mixed-Poisson distributed.
To be more precise, the conditional probability
mass function of $N$ is given by
\[
\mathbb P[N=j \, | \,\varTheta=\theta] =\frac{(\lambda\theta)^j}{j!}\, e^{-\lambda\theta},
\quad j\ge0.
\]

Hence, the unconditional probability mass function of $N$ is given by
\[
\mathbb P[N=j] =\int_0^{+\infty} \mathbb P[N=j \, | \,
\varTheta=\theta]\, \mathrm{d} F_{\varTheta}(\theta), \quad j\ge0.
\]

We introduce $m+1$ different claim types reported by the policyholder.
Each claim type induces a specific penalty
for the policyholder, which will be described below. The type of the
given claim is determined by the claim amount $C$.
We assume that all claim amounts are independent and identically
distributed, and claim amounts and claim frequencies are
mutually independent. The claims are classified according to a
multinomial scheme. Let $0<c_1^* <c_2^* <\cdots<c_m^* <\infty$.
We choose these numbers so that all claims of size less than or equal
to $c_1^*$ are considered as claims of type 0,
all claims of size from the interval $(c_1^*,c_2^*]$ are considered as
claims of type 1 and so on; finally, all claims
of size greater than $c_m^*$ are considered as claims of type $m$. Let
\begin{align*}
q_0&=\mathbb P\bigl[C\le c_1^*\bigr], \qquad
q_1=\mathbb P\bigl[c_1^*< C\le c_2^*\bigr],\\
 q_2&=\mathbb P\bigl[c_2^*< C\le c_3^*
\bigr], \quad \ldots, \quad q_m=\mathbb P\bigl[C> c_m^*
\bigr].
\end{align*}
Thus, each time a claim is reported, it is classified in one of the
$m+1$ possible categories with probabilities
$q_0$, $q_1$, $\ldots$, $q_m$. It is clear that $\sum_{i=0}^{m}q_i=1$.
Moreover, it is natural to choose the $c_i^*$ so that all $q_i$ are
strictly positive.

Denote by $N_i$ the number of claims of type $i$, $0\le i\le m$.
Therefore, for a given $\varTheta$, the random variables $N_0$, $N_1$,
$\ldots$, $N_m$ are mutually independent and
the corresponding conditional probability mass function is given by
\begin{equation}
\label{eq:1} \mathbb P[N_i=j \, | \,\varTheta=\theta] =
\frac{(\lambda\theta
q_i)^j}{j!}\, e^{-\lambda\theta q_i}, \quad j\ge0, \ 0\le i\le m.
\end{equation}

The bonus--malus scale is assumed to have $s+1$ levels numbered from 0
to $s$. A higher level number indicates a higher premium.
In particular, the policyholders who are at level 0 enjoy the maximal
bonus. A specified level is assigned to a new policyholder.
Each claim-free year is rewarded by a bonus point, i.e. the
policyholder goes one level down. Each claim type entails a specific
penalty expressed as a fixed number of levels per claim. It is natural
to assume that larger claims entail more severe penalties.

We suppose that knowing the present level and the number of claims of
each type filed during the present year suffices
to determine the level to which the policyholder is transferred. So the
bonus--malus system may be represented by a Markov chain.

Let $p_{l_0 l}(\lambda\theta;\mathbf{q})$ be the probability of moving
from level $l_0$ to level $l$ for a policyholder
with annual mean claim frequency $\lambda\theta$ and vector of
probabilities $\mathbf{q} =(q_0, q_1, \ldots, q_m)^T$, where
$0\le l_0\le s$, $0\le l\le s$ and $q_i$ is the probability that the
claim is of type $i$.
Denote by $P(\lambda\theta;\mathbf{q})$ the one-step transition matrix, i.e.
\[
P(\lambda\theta;\mathbf{q})= %
\begin{pmatrix}
p_{00}(\lambda\theta;\mathbf{q}) & p_{01}(\lambda\theta;\mathbf{q}) &
\ldots& p_{0s}(\lambda\theta;\mathbf{q}) \\
p_{10}(\lambda\theta;\mathbf{q}) & p_{11}(\lambda\theta;\mathbf{q}) &
\ldots& p_{1s}(\lambda\theta;\mathbf{q}) \\
\vdots& \vdots& \ddots& \vdots\\
p_{s0}(\lambda\theta;\mathbf{q}) & p_{s1}(\lambda\theta;\mathbf{q}) &
\ldots& p_{ss}(\lambda\theta;\mathbf{q})
\end{pmatrix} %
.
\]

Taking the $n$th power of $P(\lambda\theta;\mathbf{q})$ yields the
$n$-step transition matrix with elements
$p_{l_0 l}^{(n)}(\lambda\theta;\mathbf{q})$. Here $p_{l_0
l}^{(n)}(\lambda\theta;\mathbf{q})$ is the probability
of moving from level $l_0$ to level $l$ in $n$ transitions.

The transition matrix $P(\lambda\theta;\mathbf{q})$ is assumed to be
regular, i.e. there is some integer $n_0 \ge1$ such that
all entries of $P ((\lambda\theta;\mathbf{q}) )^{n_0}$ are
strictly positive. Consequently, the Markov chain describing
the trajectory of a policyholder with expected claim frequency $\lambda
\theta$ and vector of probabilities $\mathbf{q}$
is ergodic and thus possesses a stationary distribution
\[
\boldsymbol{\pi}(\lambda\theta;\mathbf{q})= \bigl( \pi_0(\lambda\theta
;\mathbf{q}), \pi_1(\lambda\theta;\mathbf{q}), \ldots,
\pi_s(\lambda\theta;\mathbf{q}) \bigr)^T.
\]
Here $\pi_l(\lambda\theta;\mathbf{q})$ is the stationary probability
for a policyholder with annual mean claim frequency
$\lambda\theta$ to be at level $l$, i.e.
\[
\pi_{l}(\lambda\theta;\mathbf{q}) =\lim_{n\to\infty}
p_{l_0
l}^{(n)}(\lambda\theta;\mathbf{q}).
\]

The stationary probabilities $\pi_l(\lambda\theta;\mathbf{q})$ can be
obtained directly (see \cite{DeMaPiWa2007,RoScScTe1999}).
Indeed, since the matrix $P(\lambda\theta;\mathbf{q})$ is regular, the
matrix $I-P(\lambda\theta;\mathbf{q})+E$ is invertible
and $\boldsymbol{\pi}(\lambda\theta;\mathbf{q})$ is given by
\begin{equation}
\label{eq:2} \boldsymbol{\pi}^T(\lambda\theta;\mathbf{q}) =
\mathbf{e}^T \bigl( I-P(\lambda\theta;\mathbf{q})+E
\bigr)^{-1},
\end{equation}
where $\mathbf{e}$ is a column vector of 1s, $E$ is an $(s+1)\times
(s+1)$ matrix consisting of $s+1$ column vectors $\mathbf{e}$
and $I$ is an $(s+1)\times(s+1)$ identity matrix.

Next, let $L$ be the level occupied in the scale by a randomly selected
policyholder and $\pi_l$ be the proportion of
policyholders at level $l$ once the steady state has been reached. Then
\begin{equation}
\label{eq:3} \pi_l =\mathbb P[L=l] =\int_0^{+\infty}
\pi_l(\lambda\theta;\mathbf {q})\, \mathrm{d} F_{\varTheta}(
\theta), \quad0\le l\le s.
\end{equation}

It is easily seen that $\sum_{l=0}^{s} \pi_l =1$.

We denote by $r_l$ the premium relativities associated with level $l$.
The meaning is that a policyholder occupying level $l$
pays the premium equal to $\lambda r_l \mathbb E[C]$. Note that here
and 
below we 
consider only net premiums.
In what follows, we assume that $\mathbb E[C] <\infty$ and the
distribution function $F_C$ of $C$ is continuous.

To compute the premium relativities $r_l$, we use a quadratic loss
function as proposed in \cite{No1976}
(see also \cite{DeMaPiWa2007,Le1995}). To this end, we minimize the
expected squared difference between the "true"
relative premium $\varTheta$ and the relative premium $r_L$ applicable to
this policyholder after the stationary state
has been reached, i.e. we minimize $\mathbb E  [ (\varTheta-r_L)^2
 ]$. The solution to this problem is given by
\begin{equation}
\label{eq:4} r_l=\frac{\int_0^{+\infty} \theta\pi_l(\lambda\theta;\mathbf{q})\,
\mathrm{d} F_{\varTheta}(\theta)}{
\int_0^{+\infty} \pi_l(\lambda\theta;\mathbf{q})\, \mathrm{d} F_{\varTheta
}(\theta)}
\end{equation}
(see, e.g., \cite[pp.~185--186]{DeMaPiWa2007} for the details).

\section{Varying deductibles in the bonus--malus systems with different
claim types}
\label{sec:3}

The policyholder occupying level $l$ in the bonus--malus systems
described in Section~\ref{sec:2} should pay $\lambda r_l \mathbb E[C]$.
We suppose that the relative premiums of the policyholders who are in
the bonus zone, i.e. with $r_l \le1$, are unchanged.
So they pay premiums equal to $\lambda r_l \mathbb E[C]$ and are
subject to no further penalties. Let $s_0=\min\{l\colon r_l>1\}$.
The relative premiums of the policyholders who are in the malus zone,
i.e. with $r_l>1$, or equivalently occupying level $l$
such that $s_0\le l\le s$, are softened in the following way. Instead
of paying the premium equal to $\lambda r_l \mathbb E[C]$,
the policyholder who is at level $l$ pays a reduced premium equal to
$(1-\alpha_l)\lambda r_l \mathbb E[C]$ for some specified
$\alpha_l \ge0$ depending on level $l$. We suppose that $\alpha_l$
satisfy the following assumption.

\begin{assum}
\label{assum:1}
$1\le(1-\alpha_{s_0})r_{s_0}\le(1-\alpha_{s_0+1})r_{s_0+1}\le\cdots
\le(1-\alpha_s)r_s$.
\end{assum}

Note that Assumption~\ref{assum:1} implies that the policyholder
occupying a higher level in the bonus--malus scale pays a higher
reduced premium, which is not less than the basic premium $\lambda
\mathbb E[C]$. Moreover, from Assumption~\ref{assum:1} we have
$0\le\alpha_l\le1-1/r_l$ for all $l$ such that $s_0\le l\le s$. The
case $\alpha_l=0$ corresponds to the situation when
the policyholder occupying level $l$ pays premium equal to $\lambda r_l
\mathbb E[C]$ and is subject to no further penalties.
If $\alpha_l= 1-1/r_l$, then the policyholder pays only the basic
premium $\lambda\mathbb E[C]$ and has to pay something
for claims in the future. To compensate the reduced premium, the
policyholder is subject to per claim deductible, i.e. applied
to each reported claim separately, equal to $d_{l,i}$ depending on
level $l$ occupied in the malus zone and claim type $i$.
We impose the following natural restrictions on $d_{l,i}$.

 \begingroup
 \abovedisplayskip=4.5pt
 \belowdisplayskip=4.5pt
\begin{assum}
\label{assum:2}
\begin{enumerate}
\item[(i)]
For all $l$ such that $s_0\le l\le s$, we have
\[
0\,{\le}\, d_{l,0}\,{\le}\, c_1^*, \qquad0\,{\le}\, d_{l,1}\,{\le}\,
c_1^*, \qquad0\,{\le}\, d_{l,2}\,{\le}\, c_2^*, \quad\ldots,
\quad0\,{\le}\, d_{l,m}\,{\le}\, c_m^*.
\]
\item[(ii)]
For every fixed $l$ such that $s_0\le l\le s$, we have
\[
d_{l,0}\le d_{l,1}\le\cdots\le d_{l,m}.
\]
\item[(iii)]
For every fixed $i$ such that $0\le i\le m$, we have
\[
d_{s_0,i}\le d_{s_0+1,i}\le\cdots\le d_{s,i}.
\]
\end{enumerate}
\end{assum}
\endgroup

Assertion~(i) of Assumption~\ref{assum:2} means that if a claim
reported is of type $i$, where $0\le i\le m$, then the deductible
applied to it is strictly less than the claim amount, i.e. the
insurance company covers at least some part of the losses.
Next, assertion~(ii) implies that larger claims are subject to higher
deductibles. Finally, assertion~(iii) means that
the higher level in the bonus--malus scale implies a higher deductible
for each specified claim amount.

For policyholders occupying level $l$, the deductibles $d_{l,0}$,
$d_{l,1}$, $\ldots$, $d_{l,m}$ are found using the indifference
principle (see \cite{DeMaPiWa2007,PiDeWa2005}): for this group of
policyholders, the part $\alpha_l$ of the penalties
induced by the bonus--malus system is on average equal to the total
amount of deductibles paid by these policyholders.
Thus, the indifference principle for the policyholder occupying level
$l$ can be written in the following way:\allowdisplaybreaks[0]
\begin{align}
 \lambda r_l \mathbb E[C]
&=(1-\alpha_l)\lambda r_l \mathbb E[C] +\lambda
r_l \bigl( \mathbb E[C \, | \,C\le d_{l,0}]\, \mathbb P[C
\le d_{l,0}]\notag\\[-1pt]
&\quad+d_{l,0}\, \mathbb P\bigl[d_{l,0}< C\le
c_1^*\bigr] +d_{l,1}\, \mathbb P\bigl[c_1^*< C
\le c_2^*\bigr] +\ldots\notag
\\[-1pt]
&\quad+d_{l,m}\, \mathbb P\bigl[C> c_m^*\bigr] \bigr),
\quad s_0\le l\le s. \label{eq:5}\vadjust{\goodbreak}
\end{align}%

On the left-hand side of~\eqref{eq:5}, we have the premium paid by the
policyholder occupying level $l$ in the bonus--malus system
described in Section~\ref{sec:2}. On the right-hand side of~\eqref
{eq:5}, we have the expected amount paid by this policyholder
in the bonus--malus system with varying deductibles. This amount
consists of the reduced premium and the expected amount of
penalties induced by deductibles.\allowdisplaybreaks[1]

Equation~\eqref{eq:5} can be rewritten as
\begin{equation*}
\begin{split} \alpha_l \mathbb E[C] &=\mathbb E[C \, |
\,C\le d_{l,0}]\, \mathbb P[C\le d_{l,0}] +d_{l,0}\,
\mathbb P\bigl[d_{l,0}< C\le c_1^*\bigr]
\\
&\quad+d_{l,1}\, \mathbb P\bigl[c_1^*< C\le
c_2^*\bigr] +\cdots+d_{l,m}\, \mathbb P\bigl[C>
c_m^*\bigr], \quad s_0\le l\le s, \end{split}
\end{equation*}
which is equivalent to
\begin{equation}
\label{eq:6} %
\begin{split} \alpha_l \mathbb E[C] &=
\mathbb E[C \, | \,C\le d_{l,0}]\, \mathbb P[C\le d_{l,0}]
+d_{l,0}\, \bigl( q_0-\mathbb P[C\le d_{l,0}]
\bigr)
\\
&\quad+d_{l,1} q_1 +\cdots+d_{l,m}
q_m, \quad s_0\le l\le s. \end{split} %
\end{equation}

Thus, in order to introduce varying deductibles in the bonus--malus
system, we have to choose $\alpha_l$, $d_{l,0}$, $d_{l,1}$,
$\ldots$, $d_{l,m}$ satisfying~\eqref{eq:6} and Assumptions~\ref
{assum:1} and~\ref{assum:2} for all $l$ such that $s_0\le l\le s$.
In what follows, we call any such combination of $\alpha_l$, $d_{l,0}$,
$d_{l,1}$, $\ldots$, $d_{l,m}$, where $s_0\le l\le s$,
by a solution to~\eqref{eq:6}.

\begin{lemma}
\label{lem:1}
The right-hand side of~\eqref{eq:6} is a non-decreasing function of
each of the variables $d_{l,0}$, $d_{l,1}$, $\ldots$, $d_{l,m}$.
\end{lemma}

\begin{proof}
The assertion of Lemma~\ref{lem:1} is evident for variables $d_{l,1}$,
$\ldots$, $d_{l,m}$. We now show it for $d_{l,0}$.
Let $0\le d'_{l,0}\le d''_{l,0}\le c_1^*$. Then we get
\begin{equation*}
\begin{split} &\mathbb E\bigl[C \, | \,C\le d''_{l,0}
\bigr]\, \mathbb P\bigl[C\le d''_{l,0}\bigr]
+d''_{l,0}\, \mathbb P\bigl[d''_{l,0}<
C\le c_1^*\bigr]
\\
&\qquad-\mathbb E\bigl[C \, | \,C\le d'_{l,0}\bigr]\,
\mathbb P\bigl[C\le d'_{l,0}\bigr] -d'_{l,0}
\, \mathbb P\bigl[d'_{l,0}< C\le c_1^*\bigr]
\\
&\quad\ge\frac{\mathbb E[C \, | \,C\le d'_{l,0}]\, \mathbb P[C\le
d'_{l,0}] +d'_{l,0}\, \mathbb P[d'_{l,0}< C\le d''_{l,0}]}{
\mathbb P[C\le d''_{l,0}]}\, \mathbb P\bigl[C\le d''_{l,0}
\bigr]
\\
&\qquad+d''_{l,0}\, \mathbb P
\bigl[d''_{l,0}\,{<}\, C\,{\le}\, c_1^*\bigr]
-\mathbb E\bigl[C \, | \,C\,{\le}\, d'_{l,0}\bigr]\, \mathbb P
\bigl[C\,{\le}\, d'_{l,0}\bigr] \,{-}\,d'_{l,0}\,
\mathbb P\bigl[d'_{l,0}\,{<}\, C\,{\le}\, c_1^*\bigr]
\\
&\quad=d'_{l,0}\, \mathbb P\bigl[d'_{l,0}<
C\le d''_{l,0}\bigr] +d''_{l,0}
\, \mathbb P\bigl[d''_{l,0}< C\le
c_1^*\bigr] -d'_{l,0}\, \mathbb P
\bigl[d'_{l,0}< C\le c_1^*\bigr]
\\
&\quad=\bigl(d''_{l,0} -d'_{l,0}
\bigr)\, \mathbb P\bigl[d''_{l,0}< C\le
c_1^*\bigr] \ge0, \end{split} %
\end{equation*}
which proves the lemma.
\end{proof}

In addition, taking into account the continuity of $F_C$ gives that the
right-hand side of~\eqref{eq:6} is continuous
in $d_{l,0}$, $d_{l,1}$, $\ldots$, $d_{l,m}$.

\begin{propos}
\label{propos:1}
For any fixed $\alpha_l$ satisfying Assumption~\ref{assum:1}, we have
\begin{equation}
\label{eq:7} d_{l,m}\le\min \bigl\{ \alpha_l \mathbb
E[C]/q_m, \, c_m^* \bigr\}, \quad s_0\le l
\le s.
\end{equation}
\end{propos}

\begin{proof}
From Lemma~\ref{lem:1}, we conclude that for any fixed $\alpha_l$, the
maximum possible value of $d_{l,m}$ is when
$d_{l,0} =d_{l,1} =\cdots=d_{l,m-1}=0$. Therefore, $d_{l,m}\le\alpha
_l \mathbb E[C]/q_m$.
Taking into account assertion (i) of Assumption~\ref{assum:2}
yields~\eqref{eq:7}.
\end{proof}

\begin{propos}
\label{propos:2}
If $l_1$ and $l_2$ are such that $s_0\le l_1< l_2\le s$, then $\alpha
_{l_1}\le\alpha_{l_2}$.
\end{propos}

\begin{proof}
By assertion (iii) of Assumption~\ref{assum:2}, we have $d_{l_1,0}\le
d_{l_2,0}$, $d_{l_1,1}\le d_{l_2,1}$, $\ldots$,
$d_{l_1,m}\le d_{l_2,m}$. Consequently, from Lemma~\ref{lem:1}, it
follows that\vadjust{\goodbreak}
\begin{equation*}
\begin{split} &\mathbb E[C \, | \,C\le d_{l_1,0}]\, \mathbb P[C
\le d_{l_1,0}] +d_{l_1,0}\, \bigl( q_0-\mathbb P[C\le
d_{l_1,0}] \bigr)
\\
&\qquad+d_{l_1,1} q_1 +\cdots+d_{l_1,m}
q_m
\\
&\quad\le\mathbb E[C \, | \,C\le d_{l_2,0}]\, \mathbb P[C\le
d_{l_2,0}] +d_{l_2,0}\, \bigl( q_0-\mathbb P[C\le
d_{l_2,0}] \bigr)
\\
&\qquad+d_{l_2,1} q_1 +\cdots+d_{l_2,m}
q_m. \end{split} %
\end{equation*}
By~\eqref{eq:6}, we have $\alpha_{l_1}\le\alpha_{l_2}$.
\end{proof}

\begin{thm}
\label{thm:1}
For existence of a solution to~\eqref{eq:6}, it is necessary that
\begin{equation}
\label{eq:8} \alpha_l\le\min \biggl\{ 1-\frac{1}{r_l}, \,
\frac{f(c_1^*, c_2^*,
\ldots, c_m^*)}{\mathbb E[C]} \biggr\} \quad \text{for all} \ s_0\le l\le
s,
\end{equation}
where
\[
f\bigl(c_1^*, c_2^*, \ldots, c_m^*\bigr) =
\mathbb E\bigl[C \, | \,C\le c_1^*\bigr]\, q_0
+c_1^* q_1 +\cdots+c_m^* q_m.
\]
\end{thm}

\begin{proof}
By Lemma~\ref{lem:1}, the maximum value of the right-hand side of~\eqref
{eq:6} is attained when $d_{l,0} =c_1^*$, $d_{l,1} =c_1^*$,
$d_{l,2} =c_2^*$, $\ldots$\,, $d_{l,m} =c_m^*$ and equals to $\mathbb E[C
\, | \,C\le c_1^*]\, q_0 +c_1^* q_1 +\cdots+c_m^* q_m$.
Therefore, we get
\begin{equation}
\label{eq:9} \alpha_l\le\frac{\mathbb E[C \, | \,C\le c_1^*]\, q_0 +c_1^* q_1
+\cdots+c_m^* q_m}{\mathbb E[C]}, \quad
s_0\le l\le s.
\end{equation}

Since
\begin{equation*}
\begin{split} &\mathbb E\bigl[C \, | \,C\le c_1^*\bigr]
\, q_0 +c_1^* q_1 +\cdots+c_m^*
q_m
\\
&\quad<\mathbb E\bigl[C \, | \,C\le c_1^*\bigr]\, q_0
+\mathbb E\bigl[C \, | \,c_1^*< C\le c_2^*\bigr]\,
q_1 +\cdots+\mathbb E\bigl[C \, | \,C> c_m^*\bigr]\,
q_m
\\
&\quad=\mathbb E[C], \end{split} %
\end{equation*}
the right-hand side of~\eqref{eq:9} is less than 1. Thus,
combining~\eqref{eq:9} and the inequality $\alpha_l\le1-1/r_l$,
which follows immediately from Assumption~\ref{assum:1}, gives~\eqref{eq:8}.
\end{proof}

\begin{lemma}
\label{lem:2}
Let condition~\eqref{eq:8} hold. Then for any fixed $l$ such that
$s_0\le l\le s$, we can always choose $d_{l,0}$, $d_{l,1}$,
$\ldots$\,, $d_{l,m}$ satisfying assertions \emph{(i)} and \emph{(ii)} of
Assumption~\emph{\ref{assum:2}} such that~\eqref{eq:6} is true.
\end{lemma}

The assertion of the lemma is evident.

\begin{remark}
\label{rem:1}
Note only that if
\[
\min \biggl\{ 1-\frac{1}{r_l}, \, \frac{f(c_1^*, c_2^*, \ldots,
c_m^*)}{\mathbb E[C]} \biggr\} =
\frac{f(c_1^*, c_2^*, \ldots, c_m^*)}{\mathbb E[C]},
\]
i.e.
\[
\alpha_l =\frac{f(c_1^*, c_2^*, \ldots, c_m^*)}{\mathbb E[C]},
\]
then the combination is unique: $d_{l,0}=c_1^*$, $d_{l,1}=c_1^*$,
$d_{l,2}=c_2^*$, $\ldots$, $d_{l,m}=c_m^*$. Otherwise, if
\[
\alpha_l <\frac{f(c_1^*, c_2^*, \ldots, c_m^*)}{\mathbb E[C]},
\]
then there are infinitely many such combinations.
\end{remark}

Next, note that by Lemma~\ref{lem:2}, we can choose $d_{l,0}$,
$d_{l,1}$, $\ldots$, $d_{l,m}$ for any fixed $l$,
but assertion (iii) of Assumption~\ref{assum:2} may not hold.\vadjust{\goodbreak}

The next theorem shows that we can always replace the bonus--malus
system described in Section~\ref{sec:2}
by the bonus--malus system with varying deductibles.

\begin{thm}
\label{thm:2}
There is always a solution to~\eqref{eq:6}, i.e. a combination of
$\alpha_l$, $d_{l,0}$, $d_{l,1}$, $\ldots$\,, $d_{l,m}$,
where $s_0\le l\le s$, such that Assumptions~\ref{assum:1} and~\ref
{assum:2} hold.
\end{thm}

\begin{proof}
Consider now two cases.

1) If
\[
\min \biggl\{ 1-\frac{1}{r_{s_0}}, \, \frac{f(c_1^*, c_2^*, \ldots,
c_m^*)}{\mathbb E[C]} \biggr\} =
\frac{f(c_1^*, c_2^*, \ldots, c_m^*)}{\mathbb E[C]},
\]
let
\[
\alpha_{s_0} =\frac{f(c_1^*, c_2^*, \ldots, c_m^*)}{\mathbb E[C]}.
\]

By Lemma~\ref{lem:2} and Remark~\ref{rem:1} applied to $l=s_0$, we get
$d_{s_0,0}=c_1^*$, $d_{s_0,1}=c_1^*$, $d_{s_0,2}=c_2^*$, $\ldots$\,,
$d_{s_0,m}=c_m^*$.
To satisfy Assumptions~\ref{assum:1} and~\ref{assum:2}, for all $l$
such that $s_0+1\le l\le s$, we set
$\alpha_l=\alpha_{s_0}$, $d_{l,0}=d_{s_0,0}$, $d_{l,1}=d_{s_0,1}$,
$\ldots$\,, $d_{l,m}=d_{s_0,m}$,
which proves the theorem in the first case.

2) If
\[
\min \biggl\{ 1-\frac{1}{r_{s_0}}, \, \frac{f(c_1^*, c_2^*, \ldots,
c_m^*)}{\mathbb E[C]} \biggr\} =1-
\frac{1}{r_{s_0}},
\]
let
\[
\alpha_{s_0} =1-\frac{1}{r_{s_0}}.
\]

By Lemma~\ref{lem:2} and Remark~\ref{rem:1} applied to $l=s_0$, we can
always choose
required $d_{s_0,0}$, $d_{s_0,1}$, $\ldots$, $d_{s_0,m}$
and there are infinitely many such combinations. We take any of them.
Finally, to satisfy
Assumptions~\ref{assum:1} and~\ref{assum:2}, for all $l$ such that
$s_0+1\le l\le s$, we also set
$\alpha_l=\alpha_{s_0}$, $d_{l,0}=d_{s_0,0}$, $d_{l,1}=d_{s_0,1}$,
$\ldots$\,, $d_{l,m}=d_{s_0,m}$,
which proves the theorem in the second case.
\end{proof}

The considerations given above show that the replacement of the
bonus--malus system described in Section~\ref{sec:2}
(without deductibles) by the bonus--malus system with varying
deductibles is not unique. So the insurance company
can consider different replacements and choose one that seems to be
more attractive from the policyholders' point of view.

We now consider two special cases. Another one is considered in
Section~\ref{sec:4}.

\begin{exam}
\label{exam:1}
\emph{We now suppose that deductibles are applied only to claims of
type $m$, i.e. $d_{l,i}=0$ for all $s_0\le l\le s$ and
$0\le i\le m-1$. Therefore, \eqref{eq:6} can be rewritten as}
\begin{equation}
\label{eq:10} \alpha_l \mathbb E[C] =d_{l,m}
q_m.
\end{equation}

\emph{Taking into account Theorem~\ref{thm:1}, we take any $\alpha
_{s_0}$ such that}
\[
0< \alpha_{s_0}\le\min \biggl\{ 1-\frac{1}{r_{s_0}}, \,
\frac{c_m^*
q_m}{\mathbb E[C]} \biggr\}.
\]

\emph{By~\eqref{eq:10}, $d_{s_0,m} =\alpha_{s_0} \mathbb E[C]/q_m$. To
satisfy Assumptions~\ref{assum:1} and~\ref{assum:2},
we set $\alpha_l=\alpha_{s_0}$ and $d_{l,m}=d_{s_0,m}$ for all
$s_0+1\le l\le s$.
Thus, we get a simple replacement to the bonus--malus system.}
\end{exam}

\begin{exam}
\label{exam:2}
\emph{Let $\alpha_l =1-1/r_l$ for all $s_0\le l\le s$. By Theorem~\ref
{thm:1}, if}
\[
\frac{f(c_1^*, c_2^*, \ldots, c_m^*)}{\mathbb E[C]} <1-\frac{1}{r_{s_0}},
\]
\emph{then~\eqref{eq:6} has no suitable solution. So the bonus--malus
system cannot be replaced in this way.}
\end{exam}

\section{The case where deductibles are applied only to the claims
reported by\hfill\break policyholders
occupying the highest level in the bonus--malus scale}
\label{sec:4}

We now consider the special case where $d_{l,i}=0$ for all $s_0\le l\le
s-1$ and $0\le i\le m$. This case is of great interest
when the premium relativity for policyholders who are at level $s$ is
high enough and varying deductibles are applied
to soften the bonus--malus system for such policyholders. To meet the
conditions of Assumption~\ref{assum:1}, we require that
$(1-\alpha_s)r_s\ge r_{s-1}$, which gives $\alpha_s\le1-r_{s-1}/r_s$.
So we can choose any positive $\alpha_s$ such that
\begin{equation}
\label{eq:11} \alpha_s\le\min \biggl\{ 1-\frac{r_{s-1}}{r_s}, \,
\frac{f(c_1^*,
c_2^*, \ldots, c_m^*)}{\mathbb E[C]} \biggr\}
\end{equation}
and then find $d_{s,0}$, $d_{s,1}$, $\ldots$, $d_{s,m}$ from
equation~\eqref{eq:6} applied to $l=s$, i.e. from
\begin{equation}
\label{eq:12} %
\begin{split} \alpha_s \mathbb E[C] &=
\mathbb E[C \, | \,C\le d_{s,0}]\, \mathbb P[C\le d_{s,0}]
+d_{s,0}\, \bigl( q_0-\mathbb P[C\le d_{s,0}]
\bigr)
\\
&\quad+d_{s,1} q_1 +\cdots+d_{s,m}
q_m. \end{split} %
\end{equation}

If inequality~\eqref{eq:11} is strict, then there are infinitely many
solutions to~\eqref{eq:12}.
We now consider two allocation principles for the deductibles.

The \emph{first principle} is when the deductibles are proportional to
the average claims of each type. This principle seems to be natural
and fair for policyholders, but that is not always possible, which is
easily seen from the next theorem.

\begin{thm}
\label{thm:3}
Let $\alpha_s$ be such that inequality~\eqref{eq:11} is strict and set
\begin{equation}
\label{eq:13} x_0 \!=\! \min \biggl\{\! \frac{c_1^*}{\mathbb E[C \, | \,c_1^*< C\le c_2^*]},
\frac{c_2^*}{\mathbb E[C \, | \,c_2^*< C\le c_3^*]}, \dots, \frac{c_m^*}{\mathbb E[C \, | \, C> c_m^*]} \! \biggr\}.
\end{equation}
If
\begin{align}
 \alpha_s \mathbb E[C] &>
\mathbb E \bigl[ C \, | \,C\le x_0\, \mathbb E\bigl[C \, | \,C\le
c_1^*\bigr] \bigr]\: \mathbb P \bigl[ C\le x_0\,
\mathbb E\bigl[C \, | \,C\le c_1^*\bigr] \bigr]\notag
\\
&\quad+x_0 \bigl( \mathbb E\bigl[C \, | \,C\le c_1^*
\bigr]\, \bigl( q_0- \mathbb P \bigl[ C\le x_0\, \mathbb
E\bigl[C \, | \,C\le c_1^*\bigr] \bigr] \bigr)\notag
\\
&\quad+\mathbb E\bigl[C \, | \,c_1^*< C\le c_2^*\bigr]
\, q_1 +\cdots+\mathbb E\bigl[C \, | \,C> c_m^*\bigr]\,
q_m \bigr),
\label{eq:14}
\end{align} %
then we cannot allocate deductibles proportionally to the average
claims.

Otherwise, if
\begin{align}
 \alpha_s \mathbb E[C] &\le
\mathbb E \bigl[ C \, | \,C\le x_0\, \mathbb E\bigl[C \, | \,C\le
c_1^*\bigr] \bigr]\: \mathbb P \bigl[ C\le x_0\,
\mathbb E\bigl[C \, | \,C\le c_1^*\bigr] \bigr]\notag
\\
&\quad+x_0 \bigl( \mathbb E\bigl[C \, | \,C\le c_1^*
\bigr]\, \bigl( q_0- \mathbb P \bigl[ C\le x_0\, \mathbb
E\bigl[C \, | \,C\le c_1^*\bigr] \bigr] \bigr)\notag
\\
&\quad+\mathbb E\bigl[C \, | \,c_1^*< C\le c_2^*\bigr]
\, q_1 +\cdots+\mathbb E\bigl[C \, | \,C> c_m^*\bigr]\,
q_m \bigr),
\label{eq:15} %
\end{align}
then that is possible and~\eqref{eq:12} has a unique solution of such
kind 
expressed by
\begin{equation}
\label{eq:16} %
\begin{split} d_{s,0}&=x\, \mathbb E\bigl[C
\, | \,C\le c_1^*\bigr], \qquad d_{s,1}=x\, \mathbb E\bigl[C
\, | \,c_1^*< C\le c_2^*\bigr], \quad\ldots,
\\
d_{s,m}&=x\, \mathbb E\bigl[C \, | \,C> c_m^*\bigr],
\end{split} %
\end{equation}
where $x$ is a unique positive solution to the equation
\begin{align}
 \alpha_s \mathbb E[C] &=
\mathbb E \bigl[ C \, | \,C\le x\, \mathbb E\bigl[C \, | \,C\le c_1^*
\bigr] \bigr]\: \mathbb P \bigl[ C\le x\, \mathbb E\bigl[C \, | \,C\le
c_1^*\bigr] \bigr]\notag
\\
&\quad+x \bigl( \mathbb E\bigl[C \, | \,C\le c_1^*\bigr]\, \bigl(
q_0 -\mathbb P \bigl[ C\le x\, \mathbb E\bigl[C \, | \,C\le
c_1^*\bigr] \bigr] \bigr)\notag
\\
&\quad+\mathbb E\bigl[C \, | \,c_1^*< C\le c_2^*\bigr]
\, q_1 +\cdots+\mathbb E\bigl[C \, | \,C> c_m^*\bigr]\,
q_m \bigr).
\label{eq:17}
\end{align}
\end{thm}

\begin{proof}
Let $x$ be a proportionality coefficient, i.e. we have~\eqref{eq:16}.
To meet the assertion (i) of Assumption~\ref{assum:2}, we require that
\begin{equation*}
\begin{split} &x\, \mathbb E\bigl[C \, | \,C\le c_1^*
\bigr]\le c_1^*, \qquad x\, \mathbb E\bigl[C \, |
\,c_1^*< C\le c_2^*\bigr]\le c_1^*,
\\
&x\, \mathbb E\bigl[C \, | \,c_2^*< C\le c_3^*\bigr]\le
c_2^*, \quad\ldots, \quad x\, \mathbb E\bigl[C \, | \,C>
c_m^*\bigr]\le c_m^*, \end{split} %
\end{equation*}
which is equivalent to
\[
x \le\min \biggl\{ \frac{c_1^*}{\mathbb E[C \, | \,c_1^*< C\le c_2^*]}, \frac{c_2^*}{\mathbb E[C \, | \,c_2^*< C\le c_3^*]}, \dots,
\frac{c_m^*}{\mathbb E[C \, | \, C> c_m^*]} \biggr\}.
\]

So we suppose that $x\in[0,x_0]$, where $x_0$ is given by~\eqref
{eq:13}, and substitute~\eqref{eq:16}
into~\eqref{eq:12}, which yields~\eqref{eq:17}.

The left-hand side of~\eqref{eq:17} is a positive constant. By
Lemma~\ref{lem:1}, the right-hand side of~\eqref{eq:17} is
an increasing function of $x$ on $[0,x_0]$. By the continuity of $F_C$,
it is continuous in $x$ on $[0,x_0]$.
Moreover, this function is equal to 0 as $x=0$. Therefore, if~\eqref
{eq:15} holds, then~\eqref{eq:17} has
a unique solution $x\in[0,x_0]$. Otherwise, if~\eqref{eq:14} is true,
then there is
no solution $x\in[0,x_0]$ to~\eqref{eq:17}, which completes the proof.
\end{proof}

\begin{remark}
\label{rem:2}
Equation~\eqref{eq:17} is not solvable analytically in the general
case. To find $x$, we should use numerical methods.
\end{remark}

The \emph{second principle} implies that large claims are penalized by
means of deductibles strictly and small claims are not
penalized at all (if that is possible) or at least not so strictly.
Anyway~\eqref{eq:12} must hold.
Firstly, we check if it is possible to penalize only claims of type $m$.
To this end, we set $d_{s,0} =d_{s,1} =\cdots=d_{s,m-1} =0$, $d_{s,m}
=c_m^*$ and substitute this into~\eqref{eq:12}. If
\[
\alpha_s \mathbb E[C]\le c_m^* q_m,
\]
that is possible and the desired allocation is given by
\[
d_{s,0} =d_{s,1} =\cdots=d_{s,m-1} =0 \qquad \text{and}
\qquad d_{s,m} =\alpha_s \mathbb E[C]/q_m.
\]
Otherwise, if
\[
\alpha_s \mathbb E[C]> c_m^* q_m,
\]
we also have to penalize at least claims of type $m-1$. We set $d_{s,0}
=d_{s,1} =\cdots=d_{s,m-2} =0$,
$d_{s,m-1} =c_{m-1}^*$, $d_{s,m} =c_m^*$ and substitute this into~\eqref
{eq:12}. If
\[
\alpha_s \mathbb E[C]\le c_{m-1}^* q_{m-1}
+c_m^* q_m,
\]
then the desired allocation is given by
\begin{align*}
d_{s,0} &=d_{s,1} =\cdots=d_{s,m-2} =0, \\
d_{s,m-1} &=\frac{\alpha
_s \mathbb E[C] -c_m^* q_m}{q_{m-1}} \quad \text{and} \quad d_{s,m}
=c_m^*.
\end{align*}
Otherwise, we also have to penalize at least claims of type $m-2$. So
we set $d_{s,0} =d_{s,1} =\cdots=d_{s,m-3} =0$,
$d_{s,m-2} =c_{m-2}^*$, $d_{s,m-1} =c_{m-1}^*$, $d_{s,m} =c_m^*$,
substitute this into~\eqref{eq:12} and continue
in this way until we get the desired allocation. Note that such an
allocation is always possible provided that
inequality~\eqref{eq:11} holds.

\section{Numerical illustrations}
\label{sec:5}

We now consider the scale with 4 levels (numbered from 0 to 3) and 4
claim types (numbered from 0 to 3 with probabilities
$q_0$, $q_1$, $q_2$ and $q_3$). If no claims are reported during the
current year, the policyholder moves one level down.
Each claim of type 0 is penalized by one level, each claim of type 1 is
penalized by 2 levels, each claim of type 2 or 3
is penalized by 3 levels. Note that penalties for claims of types 2 and
3 are different due to varying deductibles.
For instance, if 2 claims of type 0 are reported during the year, then
the policyholder moves 2 levels up; if 1 claim of type 0
and 1 claim of type 1 are reported, then the policyholder moves 3
levels up, i.e. goes at the highest level anyway.

The elements of the one-step transition matrix for a policyholder with
annual mean claim frequency $\lambda\theta$
and vector of probabilities $\mathbf{q} =(q_0, q_1, q_2, q_3)^T$ are
calculated using formula~\eqref{eq:1}.
Thus, the one-step transition matrix is given by
\begin{equation*}
\begin{split} &P(\lambda\theta;\mathbf{q})
\\
&= %
\begin{pmatrix}
\scriptstyle e^{-\lambda\theta} & \scriptstyle\lambda\theta q_0
e^{-\lambda\theta}
& \scriptstyle\lambda\theta q_1 e^{-\lambda\theta} +\frac{(\lambda
\theta q_0)^2}{2} e^{-\lambda\theta}
& \scriptstyle1-e^{-\lambda\theta} -\lambda\theta(q_0+q_1) e^{-\lambda
\theta} -\frac{(\lambda\theta q_0)^2}{2} e^{-\lambda\theta} \\
\scriptstyle e^{-\lambda\theta} & \scriptstyle0 & \scriptstyle\lambda
\theta q_0 e^{-\lambda\theta}
& \scriptstyle1-e^{-\lambda\theta} -\lambda\theta q_0 e^{-\lambda\theta
} \\
\scriptstyle0 & \scriptstyle e^{-\lambda\theta} & \scriptstyle0 &
\scriptstyle1-e^{-\lambda\theta} \\
\scriptstyle0 & \scriptstyle0 & \scriptstyle e^{-\lambda\theta} &
\scriptstyle1-e^{-\lambda\theta}
\end{pmatrix} %
\!\!. \end{split} %
\end{equation*}

Next, the stationary probabilities $\pi_l(\lambda\theta;\mathbf{q})$,
$0\le l\le3$, are calculated using formula~\eqref{eq:2}.
A standard computation shows that
\[
\pi_0(\lambda\theta;\mathbf{q}) =\frac{e^{-3\lambda\theta}}{\varDelta},
\]
\[
\pi_1(\lambda\theta;\mathbf{q}) =\frac{e^{-2\lambda\theta} -e^{-3\lambda
\theta}}{\varDelta},
\]
\[
\pi_2(\lambda\theta;\mathbf{q}) =\frac{e^{-\lambda\theta} -e^{-2\lambda
\theta} -\lambda\theta q_0 e^{-3\lambda\theta}}{\varDelta},
\]
\[
\pi_3(\lambda\theta;\mathbf{q}) =\frac{1-e^{-\lambda\theta} -2\lambda
\theta q_0 e^{-2\lambda\theta}
+\lambda\theta(q_0-q_1) e^{-3\lambda\theta} -\frac{(\lambda\theta
q_0)^2}{2} e^{-3\lambda\theta}}{\varDelta},
\]
where
\[
\varDelta=1-2\lambda\theta q_0 e^{-2\lambda\theta} -\lambda\theta
q_1 e^{-3\lambda\theta} -\frac{(\lambda\theta q_0)^2}{2} e^{-3\lambda\theta}.
\]

It is easily seen that
\[
e^{-3\lambda\theta}\ge0, \qquad e^{-2\lambda\theta} -e^{-3\lambda\theta
}\ge0 \quad
\text{and} \quad e^{-\lambda\theta} -e^{-2\lambda\theta} -\lambda\theta
q_0 e^{-3\lambda
\theta}\ge0
\]
for all $\lambda>0$, $\theta\ge0$ and $0<q_0<1$. Moreover, it is
evident that $\sum_{l=0}^{3} \pi_l(\lambda\theta;\mathbf{q})=1$.
So to see that $\pi_l(\lambda\theta;\mathbf{q})$, $0\le l\le3$, are
indeed probabilities, we must show that
\begin{equation}
\label{eq:18} 1-e^{-\lambda\theta} -2\lambda\theta q_0
e^{-2\lambda\theta} +\lambda\theta(q_0-q_1)
e^{-3\lambda\theta} -\frac{(\lambda\theta
q_0)^2}{2} e^{-3\lambda\theta} \ge0
\end{equation}
for all $\lambda>0$, $\theta\ge0$, $q_0>0$ and $q_1>0$ such that $q_0+q_1<1$.

It is easy to check that the minimum value of the left-hand side
of~\eqref{eq:18} with respect to $q_0\ge0$ and $q_1\ge0$
such that $q_0+q_1\le1$ is attained as $q_0=1$, $q_1=0$ and equal to
\begin{equation}
\label{eq:19} 1-e^{-\lambda\theta} -2\lambda\theta e^{-2\lambda\theta} +\lambda\theta
e^{-3\lambda\theta} -\frac{(\lambda\theta)^2}{2} e^{-3\lambda\theta}.
\end{equation}

Introduce the function
\[
h(y)= 1-e^{-y} -2y e^{-2y} +y e^{-3y} -
\frac{y^2}{2} e^{-3y}, \quad y\ge0.
\]

Taking the derivative yields
\[
h'(y)= e^{-3y} \bigl(e^y-1\bigr)^2
+2y\bigl(2e^{-2y} -e^{-3y}\bigr) +\frac{3y^2}{2}
e^{-3y}\ge0, \quad y\ge0.
\]

Therefore, $h(y)$ is non-decreasing and its minimum value is attained
as $y=0$ and equals to 0.
Consequently, the minimum value of~\eqref{eq:19} is also 0 as $\lambda
\theta=0$, which gives~\eqref{eq:18}.

In this section, we deal with exponentially distributed claim sizes
with mean $\mu>0$, i.e. the distribution function
$F_C$ of $C$ is equal to $F_C(y) =1-e^{-y/\mu}$, $y\ge0$. Thus, we have
\[
q_0 =1-e^{-c_1^*/\mu}, \qquad q_1 =e^{-c_1^*/\mu}
-e^{-c_2^*/\mu},
\]
\[
q_2 =e^{-c_2^*/\mu} -e^{-c_3^*/\mu}, \qquad q_3
=1-e^{-c_3^*/\mu};
\]
\begin{equation*}
\begin{split} \mathbb E\bigl[C \, | \,C\le c_1^*\bigr]
&=\frac{1}{q_0} \int_0^{c_1^*}
\frac
{y}{\mu}\, e^{-y/\mu}\, \mathrm{d}y =\frac{\mu(1-e^{-c_1^*/\mu}) -c_1^* e^{-c_1^*/\mu}}{1-e^{-c_1^*/\mu}}
\\
&=\mu-\frac{c_1^* e^{-c_1^*/\mu}}{1-e^{-c_1^*/\mu}}, \end{split} %
\end{equation*}
\begin{equation*}
\begin{split} \mathbb E\bigl[C \, | \,c_1^*< C\le
c_2^*\bigr] &=\frac{1}{q_1} \int_{c_1^*}^{c_2^*}
\frac{y}{\mu}\, e^{-y/\mu}\, \mathrm{d}y =\frac{(c_1^*+\mu)e^{-c_1^*/\mu} -(c_2^*+\mu)e^{-c_2^*/\mu
}}{e^{-c_1^*/\mu}-e^{-c_2^*/\mu}}
\\
&=\mu+\frac{c_1^* e^{-c_1^*/\mu} -c_2^* e^{-c_2^*/\mu}}{e^{-c_1^*/\mu
}-e^{-c_2^*/\mu}}, \end{split} %
\end{equation*}
\begin{equation*}
\begin{split} \mathbb E\bigl[C \, | \,c_2^*< C\le
c_3^*\bigr] &=\frac{1}{q_2} \int_{c_2^*}^{c_3^*}
\frac{y}{\mu}\, e^{-y/\mu}\, \mathrm{d}y =\frac{(c_2^*+\mu)e^{-c_2^*/\mu} -(c_3^*+\mu)e^{-c_3^*/\mu
}}{e^{-c_2^*/\mu}-e^{-c_3^*/\mu}}
\\
&=\mu+\frac{c_2^* e^{-c_2^*/\mu} -c_3^* e^{-c_3^*/\mu}}{e^{-c_2^*/\mu
}-e^{-c_3^*/\mu}}, \end{split} %
\end{equation*}
\[
\mathbb E\bigl[C \, | \,C> c_3^*\bigr] =\frac{1}{q_3} \int
_{c_3^*}^{+\infty} \frac
{y}{\mu}\, e^{-y/\mu}\,
\mathrm{d}y =\frac{(c_3^*+\mu) e^{-c_3^*/\mu}}{e^{-c_3^*/\mu}} =\mu+c_3^*;
\]
\begin{equation*}
\begin{split} f\bigl(c_1^*, c_2^*,
c_3^*\bigr) &=\mu-\mu e^{-c_1^*/\mu} -c_1^*
e^{-c_1^*/\mu} +c_1^* \bigl(e^{-c_1^*/\mu} -e^{-c_2^*/\mu}
\bigr)
\\
&\qquad+c_2^* \bigl(e^{-c_2^*/\mu} -e^{-c_3^*/\mu}\bigr)
+c_3^* e^{-c_3^*/\mu}
\\
&=\mu-\mu e^{-c_1^*/\mu} +\bigl(c_2^*-c_1^*\bigr)
e^{-c_2^*/\mu} +\bigl(c_3^*-c_2^*\bigr)
e^{-c_3^*/\mu}. \end{split} %
\end{equation*}

In addition, we suppose that $F_{\varTheta}(\theta) =1-e^{-\theta}$,
$\theta\ge0$.

\begin{exam}
\label{exam:3}
\emph{Let $\lambda=0.1$, $\mu=2$, $c_1^*=1$, $c_2^*=2$, $c_3^*=4$.
Then we have $q_0 \approx0.3935$, $q_1 \approx0.2387$, $q_2 \approx
0.2325$, $q_3 \approx0.1353$.
The corresponding values of $\pi_l$ and $r_l$, $0\le l\le3$, are
calculated using formulas~\eqref{eq:3} and~\eqref{eq:4},
respectively, and are given in Tables~\ref{table:1}--\ref{table:9}.}

\begin{table}[b!]
\tabcolsep=4pt
\caption{Bonus--malus system with varying deductibles for Example~\ref
{exam:3} where deductibles are applied only to claims
reported by policyholders occupying the highest level, $\alpha_3=0.05$
and the first principle is used}\label{table:1}
\begin{tabular*}{\textwidth}{@{\extracolsep{\fill
}}cD{.}{.}{1.4}D{.}{.}{1.4}D{.}{.}{5}D{.}{.}{2}D{.}{.}{1.11}D{.}{.}{4}D{.}{.}{4}D{.}{.}{4}D{.}{.}{4}@{}}
\hline
$l$ & \multicolumn{1}{l}{$\pi_l$} & \multicolumn{1}{l}{$r_l$} &
\multicolumn{1}{l}{$\lambda r_l \mathbb E[C]$} &
\multicolumn{1}{l}{$\alpha_l$} &
\multicolumn{1}{l}{$(1-\alpha_l)\lambda r_l \mathbb E[C]$} &
\multicolumn{1}{l}{$d_{l,0}$} &
\multicolumn{1}{l}{$d_{l,1}$} &
\multicolumn{1}{l}{$d_{l,2}$} &
\multicolumn{1}{l}{$d_{l,3}$} \\
\hline
3 &0.0508 &2.1844 &0.4369 &0.05 &0.4150 &0.0230 &0.0730 &0.1420
&0.3004\\
2 &0.0591 &1.8899 &0.3780 &0 &0.3780 &0 &0 &0 &0\\
1 &0.0716 &1.6543 &0.3309 &0 &0.3309 &0 &0 &0 &0\\
0 &0.8185 &0.8050 &0.1610 &0 &0.1610 &0 &0 &0 &0\\
\hline
\end{tabular*}
\end{table}

\begin{table}[b!]
\tabcolsep=4pt
\caption{Bonus--malus system with varying deductibles for Example~\ref
{exam:3} where deductibles are applied only to claims
reported by policyholders occupying the highest level, $\alpha_3=0.13$
and the first principle is used}\label{table:2}
\begin{tabular*}{\textwidth}{@{\extracolsep{\fill
}}c
lllllllll@{}}
\hline
$l$ & \multicolumn{1}{l}{$\pi_l$} & \multicolumn{1}{l}{$r_l$} &
\multicolumn{1}{l}{$\lambda r_l \mathbb E[C]$} &
\multicolumn{1}{l}{$\alpha_l$} &
\multicolumn{1}{l}{$(1-\alpha_l)\lambda r_l \mathbb E[C]$} &
\multicolumn{1}{l}{$d_{l,0}$} &
\multicolumn{1}{l}{$d_{l,1}$} &
\multicolumn{1}{l}{$d_{l,2}$} &
\multicolumn{1}{l}{$d_{l,3}$} \\
\hline
3 &0.0508 &2.1844 &0.4369 &0.13 &0.3801 &0.0598 &0.1903 &0.3699
&0.7827\\
2 &0.0591 &1.8899 &0.3780 &0 &0.3780 &0 &0 &0 &0\\
1 &0.0716 &1.6543 &0.3309 &0 &0.3309 &0 &0 &0 &0\\
0 &0.8185 &0.8050 &0.1610 &0 &0.1610 &0 &0 &0 &0\\
\hline
\end{tabular*}  
\end{table}

\emph{We first consider the case when deductibles are applied only to
the claims reported by policyholders occupying
the highest level in the bonus--malus system. By~\eqref{eq:11}, we have}
\begin{equation*}
\begin{split} \alpha_3 &\le\min \biggl\{ 1-
\frac{r_2}{r_3}, \, \frac{f(c_1^*, c_2^*,
c_3^*)}{\mu} \biggr\} \approx \biggl\{ 1-
\frac{1.8899}{2.1844}, \, \frac{1.425489}{2} \biggr\}
\\
&\approx\min\{0.1348, 0.7127\} =0.1348. \end{split} %
\end{equation*}

\emph{By~\eqref{eq:13}, we get $x_0=0.6667$. Taking $\alpha_3=0.05$ and
$\alpha_3=0.13$ shows that~\eqref{eq:15} is true in both cases.
Hence, we can allocate deductibles proportionally to the average\vadjust{\goodbreak}
claims. By~\eqref{eq:17}, the proportionality coefficient
$x=0.050066$ if $\alpha_3=0.05$ and $x=0.130443$ if $\alpha_3=0.13$.
The corresponding values of deductibles are calculated
using~\eqref{eq:16} and given in Tables~\ref{table:1} and~\ref{table:2}.}

\emph{If we apply the second principle considered in Section~\ref{sec:4}, we get $\alpha_3 \mu\le c_3^* q_3$ for both values of $\alpha_3$.
The desired allocation is presented in Tables~\ref{table:3} and~\ref{table:4}.}

\begin{table}[t!]
\caption{Bonus--malus system with varying deductibles for Example~\ref
{exam:3} where deductibles are applied only to claims
reported by policyholders occupying the highest level, $\alpha_3=0.05$
and the second principle is used}\label{table:3}
\begin{tabular*}{\textwidth}{@{\extracolsep{\fill
}}c
lllllllll@{}}
\hline
$l$ & \multicolumn{1}{l}{$\pi_l$} & \multicolumn{1}{l}{$r_l$} &
\multicolumn{1}{l}{$\lambda r_l \mathbb E[C]$} &
\multicolumn{1}{l}{$\alpha_l$} &
\multicolumn{1}{l}{$(1-\alpha_l)\lambda r_l \mathbb E[C]$} &
\multicolumn{1}{l}{$d_{l,0}$} &
\multicolumn{1}{l}{$d_{l,1}$} &
\multicolumn{1}{l}{$d_{l,2}$} &
\multicolumn{1}{l}{$d_{l,3}$} \\
\hline
3 &0.0508 &2.1844 &0.4369 &0.05 &0.4150 &0 &0 &0 &0.7389\\
2 &0.0591 &1.8899 &0.3780 &0 &0.3780 &0 &0 &0 &0\\
1 &0.0716 &1.6543 &0.3309 &0 &0.3309 &0 &0 &0 &0\\
0 &0.8185 &0.8050 &0.1610 &0 &0.1610 &0 &0 &0 &0\\
\hline
\end{tabular*} \vspace*{-2.5pt}
\end{table}

\begin{table}[t!]
\caption{Bonus--malus system with varying deductibles for Example~\ref
{exam:3} where deductibles are applied only to claims
reported by policyholders occupying the highest level, $\alpha_3=0.13$
and the second principle is used}\label{table:4}
\begin{tabular*}{\textwidth}{@{\extracolsep{\fill
}}c
lllllllll@{}}
\hline
$l$ & \multicolumn{1}{l}{$\pi_l$} & \multicolumn{1}{l}{$r_l$} &
\multicolumn{1}{l}{$\lambda r_l \mathbb E[C]$} &
\multicolumn{1}{l}{$\alpha_l$} &
\multicolumn{1}{l}{$(1-\alpha_l)\lambda r_l \mathbb E[C]$} &
\multicolumn{1}{l}{$d_{l,0}$} &
\multicolumn{1}{l}{$d_{l,1}$} &
\multicolumn{1}{l}{$d_{l,2}$} &
\multicolumn{1}{l}{$d_{l,3}$} \\
\hline
3 &0.0508 &2.1844 &0.4369 &0.05 &0.3801 &0 &0 &0 &1.9212\\
2 &0.0591 &1.8899 &0.3780 &0 &0.3780 &0 &0 &0 &0\\
1 &0.0716 &1.6543 &0.3309 &0 &0.3309 &0 &0 &0 &0\\
0 &0.8185 &0.8050 &0.1610 &0 &0.1610 &0 &0 &0 &0\\
\hline
\end{tabular*} \vspace*{-2.5pt}
\end{table}

\begin{table}[t!]
\caption{Bonus--malus system with varying deductibles for Example~\ref
{exam:3} where deductibles are applied only to claims
of type 3}\label{table:5}
\begin{tabular*}{\textwidth}{@{\extracolsep{\fill
}}c
lllllllll@{}}
\hline
$l$ & \multicolumn{1}{l}{$\pi_l$} & \multicolumn{1}{l}{$r_l$} &
\multicolumn{1}{l}{$\lambda r_l \mathbb E[C]$} &
\multicolumn{1}{l}{$\alpha_l$} &
\multicolumn{1}{l}{$(1-\alpha_l)\lambda r_l \mathbb E[C]$} &
\multicolumn{1}{l}{$d_{l,0}$} &
\multicolumn{1}{l}{$d_{l,1}$} &
\multicolumn{1}{l}{$d_{l,2}$} &
\multicolumn{1}{l}{$d_{l,3}$} \\
\hline
3 &0.0508 &2.1844 &0.4369 &0.24 &0.3320 &0 &0 &0 &3.5467\\
2 &0.0591 &1.8899 &0.3780 &0.13 &0.3288 &0 &0 &0 &1.9212\\
1 &0.0716 &1.6543 &0.3309 &0.06 &0.3110 &0 &0 &0 &0.8867\\
0 &0.8185 &0.8050 &0.1610 &0 &0.1610 &0 &0 &0 &0\\
\hline
\end{tabular*}  \vspace*{-2.5pt}
\end{table}

\begin{table}[t!]
\caption{Bonus--malus system with varying deductibles for Example~\ref{exam:3}
where deductibles are applied only to claims of type 3.
Larger values of $\alpha_l$}\label{table:6}
\begin{tabular*}{\textwidth}{@{\extracolsep{\fill
}}c
lllllllll@{}}
\hline
$l$ & \multicolumn{1}{l}{$\pi_l$} & \multicolumn{1}{l}{$r_l$} &
\multicolumn{1}{l}{$\lambda r_l \mathbb E[C]$} &
\multicolumn{1}{l}{$\alpha_l$} &
\multicolumn{1}{l}{$(1-\alpha_l)\lambda r_l \mathbb E[C]$} &
\multicolumn{1}{l}{$d_{l,0}$} &
\multicolumn{1}{l}{$d_{l,1}$} &
\multicolumn{1}{l}{$d_{l,2}$} &
\multicolumn{1}{l}{$d_{l,3}$} \\
\hline
3 &0.0508 &2.1844 &0.4369 &0.26 &0.3233 &0 &0 &0 &3.8423\\
2 &0.0591 &1.8899 &0.3780 &0.25 &0.2835 &0 &0 &0 &3.6945\\
1 &0.0716 &1.6543 &0.3309 &0.24 &0.2514 &0 &0 &0 &3.5467\\
0 &0.8185 &0.8050 &0.1610 &0 &0.1610 &0 &0 &0 &0\\
\hline
\end{tabular*}\vspace*{-2.5pt}
\end{table}

\emph{Tables~\ref{table:5} and~\ref{table:6} give examples of
bonus--malus system in the case where deductibles are applied only
to claims of type 3. In Table~\ref{table:6}, we take larger values of
$\alpha_l$, so we get higher values of the deductibles.}\vadjust{\goodbreak}

\emph{If we also apply deductibles to claims of type 2 for the same
values of $\alpha_l$, the values of deductibles $d_{l,3}$
are not so high (see Table~\ref{table:7}). Moreover, in this case, we
can take larger values of $\alpha_l$ such that
Assumptions~\ref{assum:1} and~\ref{assum:2} hold (see Table~\ref{table:8}).}

\emph{Table~\ref{table:9} presents an example of bonus--malus system in
the case where deductibles are applied to claims
of types 1, 2 and 3. On the one hand, policyholders who are in the
malus zone pay not such high premiums. On the other hand,
the corresponding values of deductibles are moderate. This bonus--malus
system seems to be more attractive.}

\begin{table}[t!]
\caption{Bonus--malus system with varying deductibles for Example~\ref
{exam:3} where deductibles are applied to claims
of types 2 and 3}\label{table:7}
\begin{tabular*}{\textwidth}{@{\extracolsep{\fill
}}c
lllllllll@{}}
\hline
$l$ & \multicolumn{1}{l}{$\pi_l$} & \multicolumn{1}{l}{$r_l$} &
\multicolumn{1}{l}{$\lambda r_l \mathbb E[C]$} &
\multicolumn{1}{l}{$\alpha_l$} &
\multicolumn{1}{l}{$(1-\alpha_l)\lambda r_l \mathbb E[C]$} &
\multicolumn{1}{l}{$d_{l,0}$} &
\multicolumn{1}{l}{$d_{l,1}$} &
\multicolumn{1}{l}{$d_{l,2}$} &
\multicolumn{1}{l}{$d_{l,3}$} \\
\hline
3 &0.0508 &2.1844 &0.4369 &0.26 &0.3233 &0 &0 &1.1 &1.9522\\
2 &0.0591 &1.8899 &0.3780 &0.25 &0.2835 &0 &0 &1.1 &1.8044\\
1 &0.0716 &1.6543 &0.3309 &0.24 &0.2514 &0 &0 &1.1 &1.6566\\
0 &0.8185 &0.8050 &0.1610 &0 &0.1610 &0 &0 &0 &0\\
\hline
\end{tabular*}  \vspace*{-3pt}
\end{table}

\begin{table}[t!]
\caption{Bonus--malus system with varying deductibles for Example~\ref
{exam:3} where deductibles are applied to claims
of types 2 and 3. Larger values of $\alpha_l$}\label{table:8}
\begin{tabular*}{\textwidth}{@{\extracolsep{\fill
}}c
lllllllll@{}}
\hline
$l$ & \multicolumn{1}{l}{$\pi_l$} & \multicolumn{1}{l}{$r_l$} &
\multicolumn{1}{l}{$\lambda r_l \mathbb E[C]$} &
\multicolumn{1}{l}{$\alpha_l$} &
\multicolumn{1}{l}{$(1-\alpha_l)\lambda r_l \mathbb E[C]$} &
\multicolumn{1}{l}{$d_{l,0}$} &
\multicolumn{1}{l}{$d_{l,1}$} &
\multicolumn{1}{l}{$d_{l,2}$} &
\multicolumn{1}{l}{$d_{l,3}$} \\
\hline
3 &0.0508 &2.1844 &0.4369 &0.45 &0.2403 &0 &0 &1.7 &3.7291\\
2 &0.0591 &1.8899 &0.3780 &0.40 &0.2268 &0 &0 &1.6 &3.1620\\
1 &0.0716 &1.6543 &0.3309 &0.35 &0.2151 &0 &0 &1.5 &2.5949\\
0 &0.8185 &0.8050 &0.1610 &0 &0.1610 &0 &0 &0 &0\\
\hline
\end{tabular*}    \vspace*{-3pt}
\end{table}

\begin{table}[t!]
\caption{Bonus--malus system with varying deductibles for Example~\ref
{exam:3} where deductibles are applied to claims
of types 1, 2 and 3}\label{table:9}
\begin{tabular*}{\textwidth}{@{\extracolsep{\fill
}}c
lllllllll@{}}
\hline
$l$ & \multicolumn{1}{l}{$\pi_l$} & \multicolumn{1}{l}{$r_l$} &
\multicolumn{1}{l}{$\lambda r_l \mathbb E[C]$} &
\multicolumn{1}{l}{$\alpha_l$} &
\multicolumn{1}{l}{$(1-\alpha_l)\lambda r_l \mathbb E[C]$} &
\multicolumn{1}{l}{$d_{l,0}$} &
\multicolumn{1}{l}{$d_{l,1}$} &
\multicolumn{1}{l}{$d_{l,2}$} &
\multicolumn{1}{l}{$d_{l,3}$} \\
\hline
3 &0.0508 &2.1844 &0.4369 &0.45 &0.2403 &0 &0.7 &1.5 &2.8383\\
2 &0.0591 &1.8899 &0.3780 &0.40 &0.2268 &0 &0.5 &1.4 &2.6239\\
1 &0.0716 &1.6543 &0.3309 &0.35 &0.2151 &0 &0.3 &1.3 &2.4096\\
0 &0.8185 &0.8050 &0.1610 &0 &0.1610 &0 &0 &0 &0\\
\hline
\end{tabular*}  \vspace*{-3pt}
\end{table}
\end{exam}

\begin{exam}
\label{exam:4}
\emph{Let $\lambda=0.1$, $\mu=2$, $c_1^*=0.3$, $c_2^*=1.2$, $c_3^*=2.8$.
Then we have $q_0 \approx0.1393$, $q_1 \approx0.3119$, $q_2 \approx
0.3022$, $q_3 \approx0.2466$.
The corresponding values of $\pi_l$ and $r_l$ as well as examples of
bonus--malus systems with varying deductibles
are given in Tables~\ref{table:10}--\ref{table:12}.}

\begin{table}[t!]
\caption{Bonus--malus system with varying deductibles for Example~\ref
{exam:4} where deductibles are applied to claims
of types 2 and 3}\label{table:10}
\begin{tabular*}{\textwidth}{@{\extracolsep{\fill
}}c
lllllllll@{}}
\hline
$l$ & \multicolumn{1}{l}{$\pi_l$} & \multicolumn{1}{l}{$r_l$} &
\multicolumn{1}{l}{$\lambda r_l \mathbb E[C]$} &
\multicolumn{1}{l}{$\alpha_l$} &
\multicolumn{1}{l}{$(1-\alpha_l)\lambda r_l \mathbb E[C]$} &
\multicolumn{1}{l}{$d_{l,0}$} &
\multicolumn{1}{l}{$d_{l,1}$} &
\multicolumn{1}{l}{$d_{l,2}$} &
\multicolumn{1}{l}{$d_{l,3}$} \\
\hline
3 &0.0653 &2.0731 &0.4146 &0.20 &0.3317 &0 &0 &0.30 &1.2544\\
2 &0.0717 &1.7925 &0.3585 &0.15 &0.3047 &0 &0 &0.25 &0.9102\\
1 &0.0679 &1.6263 &0.3253 &0.10 &0.2927 &0 &0 &0.20 &0.5659\\
0 &0.7951 &0.7869 &0.1574 &0 &0.1574 &0 &0 &0 &0\\
\hline
\end{tabular*}  \vspace*{-3pt}
\end{table}

\begin{table}[t!]
\caption{Bonus--malus system with varying deductibles for Example~\ref
{exam:4} where deductibles are applied to claims
of types 1, 2 and 3}\label{table:11}
\begin{tabular*}{\textwidth}{@{\extracolsep{\fill
}}c
lllllllll@{}}
\hline
$l$ & \multicolumn{1}{l}{$\pi_l$} & \multicolumn{1}{l}{$r_l$} &
\multicolumn{1}{l}{$\lambda r_l \mathbb E[C]$} &
\multicolumn{1}{l}{$\alpha_l$} &
\multicolumn{1}{l}{$(1-\alpha_l)\lambda r_l \mathbb E[C]$} &
\multicolumn{1}{l}{$d_{l,0}$} &
\multicolumn{1}{l}{$d_{l,1}$} &
\multicolumn{1}{l}{$d_{l,2}$} &
\multicolumn{1}{l}{$d_{l,3}$} \\
\hline
3 &0.0653 &2.0731 &0.4146 &0.24 &0.3151 &0 &0.10 &0.60 &1.0847\\
2 &0.0717 &1.7925 &0.3585 &0.22 &0.2796 &0 &0.10 &0.55 &0.9838\\
1 &0.0679 &1.6263 &0.3253 &0.20 &0.2602 &0 &0.05 &0.50 &0.9461\\
0 &0.7951 &0.7869 &0.1574 &0 &0.1574 &0 &0 &0 &0\\
\hline
\end{tabular*}
\end{table}

\begin{table}[t!]
\caption{Bonus--malus system with varying deductibles for Example~\ref
{exam:4} where deductibles are applied to claims
of types 1, 2 and 3. Larger values of $\alpha_l$}\label{table:12}
\begin{tabular*}{\textwidth}{@{\extracolsep{\fill
}}c
lllllllll@{}}
\hline
$l$ & \multicolumn{1}{l}{$\pi_l$} & \multicolumn{1}{l}{$r_l$} &
\multicolumn{1}{l}{$\lambda r_l \mathbb E[C]$} &
\multicolumn{1}{l}{$\alpha_l$} &
\multicolumn{1}{l}{$(1-\alpha_l)\lambda r_l \mathbb E[C]$} &
\multicolumn{1}{l}{$d_{l,0}$} &
\multicolumn{1}{l}{$d_{l,1}$} &
\multicolumn{1}{l}{$d_{l,2}$} &
\multicolumn{1}{l}{$d_{l,3}$} \\
\hline
3 &0.0653 &2.0731 &0.4146 &0.45 &0.2280 &0 &0.20 &0.9 &2.2937\\
2 &0.0717 &1.7925 &0.3585 &0.40 &0.2151 &0 &0.15 &0.8 &2.0740\\
1 &0.0679 &1.6263 &0.3253 &0.35 &0.2114 &0 &0.10 &0.7 &1.8543\\
0 &0.7951 &0.7869 &0.1574 &0 &0.1574 &0 &0 &0 &0\\
\hline
\end{tabular*}
\end{table}
\end{exam}

\section{Conclusion}
\label{sec:6}

The bonus--malus systems with different claim types and varying
deductibles eliminate both drawbacks of the traditional
bonus--malus systems mentioned in Section~\ref{sec:1} and present a
number of advantages, namely:
\begin{itemize}
\item
Policyholders reporting small and large claims are not penalized in the
same way. This helps to avoid or at least decrease bonus hunger.

\item
Policyholders will do all their best to prevent or at least decrease
the losses.

\item
Even if a policyholder leaves the company after a claim, he has to pay
for the deductible.

\item
Relative premiums and amounts of deductibles may be tuned in an optimal
way in order to attract policyholders.
\end{itemize}

Section~\ref{sec:5} gives a few examples of such bonus--malus systems.
The numerical illustrations show that use of both penalty types
(premium relativities and varying deductibles) in this way seems indeed
attractive and fair for policyholders.
On the one hand, policyholders who are in the malus zone pay not such
high premiums.
On the other hand, the corresponding values of deductibles are moderate.
In addition, it is fair that only policyholders reporting (large)
claims are subject to further penalties, i.e. deductibles.

\section*{Acknowledgments}
The author is deeply grateful to the anonymous referees for careful
reading and valuable comments and suggestions,
which helped to improve the earlier version of the paper.


\begin{thebibliography}{99}

\bibitem{Bo2005}
%
\begin{barticle}
\bauthor{\bsnm{Bonsdorff}, \binits{H.}}:
\batitle{On asymptotic properties of {b}onus--{m}alus systems based on the
number and on the size of the claims}.
\bjtitle{Scand. Actuar. J.}
\bvolume{2005},
\bfpage{309}--\blpage{320}
(\byear{2005}).
\bid{doi={10.1080/03461230510009826}, mr={2164049}}
\end{barticle}
%
%
\OrigBibText
%
\begin{barticle}
\bauthor{\bsnm{Bonsdorff}, \binits{H.}}:
\batitle{On asymptotic preperties of {B}onus--{M}alus systems based on the
number and on the size of the claims}.
\bjtitle{Scandinavian Actuarial Journal}
\bvolume{2005},
\bfpage{309}--\blpage{320}
(\byear{2005})
\end{barticle}
%
\endOrigBibText
\bptok{structpyb}%
\endbibitem

\bibitem{DeDh2001}
%
\begin{barticle}
\bauthor{\bsnm{Denuit}, \binits{M.}},
\bauthor{\bsnm{Dhaene}, \binits{J.}}:
\batitle{{B}onus--{m}alus scales using exponential loss functions}.
\bjtitle{Bl{\"a}tter DGVFM}
\bvolume{25},
\bfpage{13}--\blpage{27}
(\byear{2001})
\end{barticle}
%
%
\OrigBibText
%
\begin{barticle}
\bauthor{\bsnm{Denuit}, \binits{M.}},
\bauthor{\bsnm{Dhaene}, \binits{J.}}:
\batitle{{B}onus--{M}alus scales using exponential loss functions}.
\bjtitle{Bl{\"a}tter der DGVFM}
\bvolume{25},
\bfpage{13}--\blpage{27}
(\byear{2001})
\end{barticle}
%
\endOrigBibText
\bptok{structpyb}%
\endbibitem

\bibitem{DeMaPiWa2007}
%
\begin{bbook}
\bauthor{\bsnm{Denuit}, \binits{M.}},
\bauthor{\bsnm{Mar{\'e}chal}, \binits{X.}},
\bauthor{\bsnm{Pitrebois}, \binits{S.}},
\bauthor{\bsnm{Walhin}, \binits{J.-F.}}:
\bbtitle{Actuarial Modelling of Claim Counts: Risk Classification, Credibility
and Bonus--Malus Systems}.
\bpublisher{John Wiley \& Sons},
\blocation{Chichester}
(\byear{2007}).
\bid{doi={10.1002/9780470517420}, mr={2384837}}
\end{bbook}
%
%
\OrigBibText
%
\begin{bbook}
\bauthor{\bsnm{Denuit}, \binits{M.}},
\bauthor{\bsnm{Mar{\'e}chal}, \binits{X.}},
\bauthor{\bsnm{Pitrebois}, \binits{S.}},
\bauthor{\bsnm{Walhin}, \binits{J.-F.}}:
\bbtitle{Actuarial Modelling of Claim Counts: Risk Classification, Credibility
and Bonus--Malus Systems}.
\bpublisher{John Wiley \& Sons},
\blocation{Chichester}
(\byear{2007})
\end{bbook}
%
\endOrigBibText
\bptok{structpyb}%
\endbibitem

\bibitem{DiVa1989}
%
\begin{barticle}
\bauthor{\bsnm{Dionne}, \binits{G.}},
\bauthor{\bsnm{Vanasse}, \binits{C.}}:
\batitle{A generalization of actuarial automobile insurance rating
models: the
negative binomial distribution with a regression component}.
\bjtitle{ASTIN Bull.}
\bvolume{19},
\bfpage{199}--\blpage{212}
(\byear{1989})
\end{barticle}
%
%
\OrigBibText
%
\begin{barticle}
\bauthor{\bsnm{Dionne}, \binits{G.}},
\bauthor{\bsnm{Vanasse}, \binits{C.}}:
\batitle{A generalization of actuarial automobile insurance rating
models: the
negative binomial distribution with a regression component}.
\bjtitle{ASTIN Bulletin}
\bvolume{19},
\bfpage{199}--\blpage{212}
(\byear{1989})
\end{barticle}
%
\endOrigBibText
\bptok{structpyb}%
\endbibitem

\bibitem{DiVa1992}
%
\begin{barticle}
\bauthor{\bsnm{Dionne}, \binits{G.}},
\bauthor{\bsnm{Vanasse}, \binits{C.}}:
\batitle{Automobile insurance ratemaking in the presence of asymmetrical
information}.
\bjtitle{J. Appl. Econom.}
\bvolume{7},
\bfpage{149}--\blpage{165}
(\byear{1992})
\end{barticle}
%
%
\OrigBibText
%
\begin{barticle}
\bauthor{\bsnm{Dionne}, \binits{G.}},
\bauthor{\bsnm{Vanasse}, \binits{C.}}:
\batitle{Automobile insurance ratemaking in the presence of asymmetrical
information}.
\bjtitle{Journal of Applied Econometrics}
\bvolume{7},
\bfpage{149}--\blpage{165}
(\byear{1992})
\end{barticle}
%
\endOrigBibText
\bptok{structpyb}%
\endbibitem

\bibitem{FrVr2001}
%
\begin{barticle}
\bauthor{\bsnm{Frangos}, \binits{N.E.}},
\bauthor{\bsnm{Vrontos}, \binits{S.D.}}:
\batitle{Design of optimal bonus--malus systems with a frequency and a severity
component on an individual basis in automobile insurance}.
\bjtitle{ASTIN Bull.}
\bvolume{31},
\bfpage{1}--\blpage{22}
(\byear{2001}).
\bid{doi={10.2143/AST.31.1.991}, mr={1945629}}
\end{barticle}
%
%
\OrigBibText
%
\begin{barticle}
\bauthor{\bsnm{Frangos}, \binits{N.E.}},
\bauthor{\bsnm{Vrontos}, \binits{S.D.}}:
\batitle{Design of optimal bonus--malus systems with a frequency and a severity
component on an individual basis in automobile insurance}.
\bjtitle{ASTIN Bulletin}
\bvolume{31},
\bfpage{1}--\blpage{22}
(\byear{2001})
\end{barticle}
%
\endOrigBibText
\bptok{structpyb}%
\endbibitem

\bibitem{GoDe2016}
%
\begin{barticle}
\bauthor{\bsnm{G{\'o}mez-D{\'e}niz}, \binits{E.}}:
\batitle{Bivariate credibility bonus--malus premiums distinguishing
between two
types of claims}.
\bjtitle{Insur. Math. Econ.}
\bvolume{70},
\bfpage{117}--\blpage{124}
(\byear{2016}).
\bid{doi={10.1016/\\j.insmatheco.2016.06.009}, mr={3543037}}
\end{barticle}
%
%
\OrigBibText
%
\begin{barticle}
\bauthor{\bsnm{G{\'o}mez-D{\'e}niz}, \binits{E.}}:
\batitle{Bivariate credibility bonus--malus premiums distinguishing
between two
types of claims}.
\bjtitle{Insurance: Mathematics and Economics}
\bvolume{70},
\bfpage{117}--\blpage{124}
(\byear{2016})
\end{barticle}
%
\endOrigBibText
\bptok{structpyb}%
\endbibitem

\bibitem{GoDeHeBaFeSa2014}
%
\begin{barticle}
\bauthor{\bsnm{G{\'o}mez-D{\'e}niz}, \binits{E.}},
\bauthor{\bsnm{Hern{\'a}ndez-Bastida}, \binits{A.}},
\bauthor{\bsnm{Fern{\'a}ndez-S{\'a}nchez}, \binits{M.P.}}:
\batitle{Computing credibility bonus--malus premiums using the total claim
amount distribution}.
\bjtitle{Hacet. J. Math. Stat.}
\bvolume{43},
\bfpage{1047}--\blpage{1061}
(\byear{2014}).
\bid{mr={3331161}}
\end{barticle}
%
%
\OrigBibText
%
\begin{barticle}
\bauthor{\bsnm{G{\'o}mez-D{\'e}niz}, \binits{E.}},
\bauthor{\bsnm{Hern{\'a}ndez-Bastida}, \binits{A.}},
\bauthor{\bsnm{Fern{\'a}ndez-S{\'a}nchez}, \binits{M.P.}}:
\batitle{Computing credibility bonus--malus premiums using the total claim
amount distribution}.
\bjtitle{Hacettepe Journal of Mathematics and Statistics}
\bvolume{43},
\bfpage{1047}--\blpage{1061}
(\byear{2014})
\end{barticle}
%
\endOrigBibText
\bptok{structpyb}%
\endbibitem

\bibitem{Ho1994}
%
\begin{barticle}
\bauthor{\bsnm{Holtan}, \binits{J.}}:
\batitle{Bonus made easy}.
\bjtitle{ASTIN Bull.}
\bvolume{24},
\bfpage{61}--\blpage{74}
(\byear{1994})
\end{barticle}
%
%
\OrigBibText
%
\begin{barticle}
\bauthor{\bsnm{Holtan}, \binits{J.}}:
\batitle{Bonus made easy}.
\bjtitle{ASTIN Bulletin}
\bvolume{24},
\bfpage{61}--\blpage{74}
(\byear{1994})
\end{barticle}
%
\endOrigBibText
\bptok{structpyb}%
\endbibitem

\bibitem{Le1995}
%
\begin{bbook}
\bauthor{\bsnm{Lemaire}, \binits{J.}}:
\bbtitle{Bonus--Malus Systems in Automobile Insurance}.
\bpublisher{Kluwer Academic Publisher},
\blocation{Boston}
(\byear{1995})
\end{bbook}
%
%
\OrigBibText
%
\begin{bbook}
\bauthor{\bsnm{Lemaire}, \binits{J.}}:
\bbtitle{Bonus--Malus Systems in Automobile Insurance}.
\bpublisher{Kluwer Academic Publisher},
\blocation{Boston}
(\byear{1995})
\end{bbook}
%
\endOrigBibText
\bptok{structpyb}%
\endbibitem

\bibitem{LeZi1994}
%
\begin{barticle}
\bauthor{\bsnm{Lemaire}, \binits{J.}},
\bauthor{\bsnm{Zi}, \binits{H.}}:
\batitle{High deductibles instead of {b}onus--{m}alus. {C}an it work?}
\bjtitle{ASTIN Bull.}
\bvolume{24},
\bfpage{75}--\blpage{88}
(\byear{1994})
\end{barticle}
%
%
\OrigBibText
%
\begin{barticle}
\bauthor{\bsnm{Lemaire}, \binits{J.}},
\bauthor{\bsnm{Zi}, \binits{H.}}:
\batitle{High deductibles instead of {B}onus--{M}alus. {C}an it work?}
\bjtitle{ASTIN Bulletin}
\bvolume{24},
\bfpage{75}--\blpage{88}
(\byear{1994})
\end{barticle}
%
\endOrigBibText
\bptok{structpyb}%
\endbibitem

\bibitem{MaHa2009}
%
\begin{barticle}
\bauthor{\bsnm{Mahmoudvand}, \binits{R.}},
\bauthor{\bsnm{Hassani}, \binits{H.}}:
\batitle{Generalized bonus--malus systems with a frequency and a severity
component on an individual basis in automobile insurance}.
\bjtitle{ASTIN Bull.}
\bvolume{39},
\bfpage{307}--\blpage{315}
(\byear{2009}).
\bid{doi={10.2143/AST.39.1.2038066}, mr={2749888}}
\end{barticle}
%
%
\OrigBibText
%
\begin{barticle}
\bauthor{\bsnm{Mahmoudvand}, \binits{R.}},
\bauthor{\bsnm{Hassani}, \binits{H.}}:
\batitle{Generalized bonus--malus systems with a frequency and a severity
component on an individual basis in automobile insurance}.
\bjtitle{ASTIN Bulletin}
\bvolume{39},
\bfpage{307}--\blpage{315}
(\byear{2009})
\end{barticle}
%
\endOrigBibText
\bptok{structpyb}%
\endbibitem

\bibitem{MeSa2005}
%
\begin{barticle}
\bauthor{\bsnm{Mert}, \binits{M.}},
\bauthor{\bsnm{Saykan}, \binits{Y.}}:
\batitle{On a bonus--malus system where the claim frequency distribution is
geometric and the claim severity distribution is Pareto}.
\bjtitle{Hacet. J. Math. Stat.}
\bvolume{34},
\bfpage{75}--\blpage{81}
(\byear{2005}).
\bid{mr={2212712}}
\end{barticle}
%
%
\OrigBibText
%
\begin{barticle}
\bauthor{\bsnm{Mert}, \binits{M.}},
\bauthor{\bsnm{Saykan}, \binits{Y.}}:
\batitle{On a bonus--malus system where the claim frequency distribution is
geometric and the claim severity distribution is pareto}.
\bjtitle{Hacettepe Journal of Mathematics and Statistics}
\bvolume{34},
\bfpage{75}--\blpage{81}
(\byear{2005})
\end{barticle}
%
\endOrigBibText
\bptok{structpyb}%
\endbibitem

\bibitem{No1976}
%
\begin{barticle}
\bauthor{\bsnm{Norberg}, \binits{R.}}:
\batitle{A credibility theory for automobile bonus system}.
\bjtitle{Scand. Actuar. J.}
\bvolume{1976},
\bfpage{92}--\blpage{107}
(\byear{1976}).
\bid{doi={10.1080/03461238.1976.10405605}, mr={0428666}}
\end{barticle}
%
%
\OrigBibText
%
\begin{barticle}
\bauthor{\bsnm{Norberg}, \binits{R.}}:
\batitle{A credibility theory for automobile bonus system}.
\bjtitle{Scandinavian Actuarial Journal}
\bvolume{1976},
\bfpage{92}--\blpage{107}
(\byear{1976})
\end{barticle}
%
\endOrigBibText
\bptok{structpyb}%
\endbibitem

\bibitem{Pi1997}
%
\begin{barticle}
\bauthor{\bsnm{Pinquet}, \binits{J.}}:
\batitle{Allowance for cost of claims in bonus--malus systems}.
\bjtitle{ASTIN Bull.}
\bvolume{27},
\bfpage{33}--\blpage{57}
(\byear{1997})
\end{barticle}
%
%
\OrigBibText
%
\begin{barticle}
\bauthor{\bsnm{Pinquet}, \binits{J.}}:
\batitle{Allowance for cost of claims in bonus--malus systems}.
\bjtitle{ASTIN Bulletin}
\bvolume{27},
\bfpage{33}--\blpage{57}
(\byear{1997})
\end{barticle}
%
\endOrigBibText
\bptok{structpyb}%
\endbibitem

\bibitem{Pi1998}
%
\begin{barticle}
\bauthor{\bsnm{Pinquet}, \binits{J.}}:
\batitle{Designing optimal bonus--malus systems from different types of claims}.
\bjtitle{ASTIN Bull.}
\bvolume{28},
\bfpage{205}--\blpage{220}
(\byear{1998})
\end{barticle}
%
%
\OrigBibText
%
\begin{barticle}
\bauthor{\bsnm{Pinquet}, \binits{J.}}:
\batitle{Designing optimal bonus--malus systems from different types of claims}.
\bjtitle{ASTIN Bulletin}
\bvolume{28},
\bfpage{205}--\blpage{220}
(\byear{1998})
\end{barticle}
%
\endOrigBibText
\bptok{structpyb}%
\endbibitem

\bibitem{PiDeWa2005}
%
\begin{barticle}
\bauthor{\bsnm{Pitrebois}, \binits{S.}},
\bauthor{\bsnm{Denuit}, \binits{M.}},
\bauthor{\bsnm{Walhin}, \binits{J.-F.}}:
\batitle{Bonus--malus systems with varying deductibles}.
\bjtitle{ASTIN Bull.}
\bvolume{35},
\bfpage{261}--\blpage{274}
(\byear{2005}).
\bid{doi={10.2143/AST.35.1.583175}, mr={2143217}}
\end{barticle}
%
%
\OrigBibText
%
\begin{barticle}
\bauthor{\bsnm{Pitrebois}, \binits{S.}},
\bauthor{\bsnm{Denuit}, \binits{M.}},
\bauthor{\bsnm{Walhin}, \binits{J.-F.}}:
\batitle{Bonus--malus systems with varying deductibles}.
\bjtitle{ASTIN Bulletin}
\bvolume{35},
\bfpage{261}--\blpage{274}
(\byear{2005})
\end{barticle}
%
\endOrigBibText
\bptok{structpyb}%
\endbibitem

\bibitem{PiDeWa2006}
%
\begin{barticle}
\bauthor{\bsnm{Pitrebois}, \binits{S.}},
\bauthor{\bsnm{Denuit}, \binits{M.}},
\bauthor{\bsnm{Walhin}, \binits{J.-F.}}:
\batitle{Multi-event bonus--malus scales}.
\bjtitle{J. Risk Insur.}
\bvolume{73},
\bfpage{517}--\blpage{528}
(\byear{2006})
\end{barticle}
%
%
\OrigBibText
%
\begin{barticle}
\bauthor{\bsnm{Pitrebois}, \binits{S.}},
\bauthor{\bsnm{Denuit}, \binits{M.}},
\bauthor{\bsnm{Walhin}, \binits{J.-F.}}:
\batitle{Multi-event bonus--malus scales}.
\bjtitle{Journal of Risk and Insurance}
\bvolume{73},
\bfpage{517}--\blpage{528}
(\byear{2006})
\end{barticle}
%
\endOrigBibText
\bptok{structpyb}%
\endbibitem

\bibitem{RoScScTe1999}
%
\begin{bbook}
\bauthor{\bsnm{Rolski}, \binits{T.}},
\bauthor{\bsnm{Schmidli}, \binits{H.}},
\bauthor{\bsnm{Schmidt}, \binits{V.}},
\bauthor{\bsnm{Teugels}, \binits{J.}}:
\bbtitle{Stochastic Processes for Insurance and Finance}.
\bpublisher{John Wiley \& Sons},
\blocation{Chichester}
(\byear{1999}).
\bid{doi={10.1002/\\9780470317044}, mr={1680267}}
\end{bbook}
%
%
\OrigBibText
%
\begin{bbook}
\bauthor{\bsnm{Rolski}, \binits{T.}},
\bauthor{\bsnm{Schmidli}, \binits{H.}},
\bauthor{\bsnm{Schmidt}, \binits{V.}},
\bauthor{\bsnm{Teugels}, \binits{J.}}:
\bbtitle{Stochastic Processes for Insurance and Finance}.
\bpublisher{John Wiley \& Sons},
\blocation{Chichester}
(\byear{1999})
\end{bbook}
%
\endOrigBibText
\bptok{structpyb}%
\endbibitem

\bibitem{TzVrFr2014}
%
\begin{barticle}
\bauthor{\bsnm{Tzougas}, \binits{G.}},
\bauthor{\bsnm{Vrontos}, \binits{S.}},
\bauthor{\bsnm{Frangos}, \binits{N.}}:
\batitle{Optimal bonus--malus systems using finite mixture models}.
\bjtitle{ASTIN Bull.}
\bvolume{44},
\bfpage{417}--\blpage{444}
(\byear{2014}).
\bid{doi={10.1017/asb.2013.31}, mr={3389581}}
\end{barticle}
%
%
\OrigBibText
%
\begin{barticle}
\bauthor{\bsnm{Tzougas}, \binits{G.}},
\bauthor{\bsnm{Vrontos}, \binits{S.}},
\bauthor{\bsnm{Frangos}, \binits{N.}}:
\batitle{Optimal bonus--malus systems using finite mixture models}.
\bjtitle{ASTIN Bulletin}
\bvolume{44},
\bfpage{417}--\blpage{444}
(\byear{2014})
\end{barticle}
%
\endOrigBibText
\bptok{structpyb}%
\endbibitem

\end{thebibliography}
\end{document}